\documentclass[11pt]{article}  
\def\final{1}

\usepackage{color}
\usepackage{xspace,fullpage,epsfig,wrapfig,paralist}
\usepackage{url}
\usepackage{graphicx}
\usepackage{paralist}
\usepackage{amsmath, amsthm, amssymb}
\usepackage{latexsym}
\usepackage{amsfonts}
\usepackage{hyperref}
\hypersetup{colorlinks=true,linkcolor=blue,citecolor=blue}
\usepackage{times}
\newcommand{\mypar}[1]{\paragraph{#1}}
\newenvironment{remark}{\smallskip\noindent {\em Remark:}}{\smallskip}

\def \N {\mathbb{N}}
\def \d {\delta}
\def \R {\mathbb{R}}

\def \E {\mathbb{E}}

\def \T {\mathbb{T}}

\def \NN {\mathcal{N}}

\def \b {\mathbf{b}}

\def \D {U}
\def \T {T}
\def \Phi {V}
\def \M {\mathbf{U}}
\def \BT {\mathbf{T}}
\def \Z {\mathbf{V}}

\def \y {\mathbf y}

\def \x {\mathbf x}

\def \z {\mathbf z}
\def \v {\mathbf v}
\def \s {{\mathbf s}}

\def \etc {,\ldots,}

\newcommand{\AAA}{\mathcal A}
\newcommand{\LIN}{\mathcal{L}_{\mathrm{lin}}}
\newcommand{\LOG}{\mathcal{L}_{\mathrm{log}}}

\newcommand{\norm}[1]{\left \| #1 \right \|}
\DeclareSymbolFont{AMSb}{U}{msb}{m}{n}
\DeclareMathSymbol{\erert}{\mathbin}{AMSb}{"50}
\DeclareMathSymbol{\C}{\mathbin}{AMSb}{"43}

\newcommand{\bit}[1]{\ensuremath{\{0,1\}^{#1}}}
\newcommand{\eg}{\emph{e.g.,}\xspace}
\ifnum\final=0
\newcommand{\mynote}[1]{\marginpar{\tiny\sf #1}}
\else
\newcommand{\mynote}[1]{}
\fi

\newenvironment{CompactEnumerate}{
\begin{list}{\arabic{enumi}.}{%
\usecounter{enumi}
\setlength{\leftmargin}{12pt}
\setlength{\itemindent}{3pt}
\setlength{\topsep}{3pt}
\setlength{\itemsep}{1pt}
}}
{\end{list}}
\newenvironment{CompactItemize}{
\begin{list}{$\bullet$}{%
\setlength{\leftmargin}{12pt}
\setlength{\itemindent}{3pt}
\setlength{\topsep}{3pt}
\setlength{\itemsep}{1pt}
}}
{\end{list}}
\newtheorem{theorem}{Theorem}[section]
\newtheorem{lemma}[theorem]{Lemma}
\newtheorem{definition}{Definition}
\newtheorem{claim}{Claim}
\newtheorem{proposition}[theorem]{Proposition}

\newtheorem{corollary}[theorem]{Corollary}

\newcommand{\anote}[1]{\mynote{AS: {#1}}}

\begin{document}
\title{\Large The Power of Linear Reconstruction Attacks}
\author{Shiva Prasad Kasiviswanathan\thanks{GE Global Research Center, San Ramon. Most of the work was done when the author was a postdoc at the IBM T. J.\ Watson Research Center, kasivisw@gmail.com.} \and \ \ \ \ Mark Rudelson\thanks{University of Michigan, rudelson@umich.edu.}  \and \ \ \ \ Adam Smith\thanks{Pennsylvania State University, asmith@cse.psu.edu.}
}
\date{}
\maketitle
\begin{abstract}
We consider the power of ``linear reconstruction attacks'' in statistical data privacy, showing that they can be applied to a much wider range of settings than previously understood. Linear attacks have been studied before~\cite{DiNi03,DMT07,KRSU10,De11,MuthuN12} but have so far been applied only in settings with releases that are ``obviously'' linear. 

Consider a database curator who manages a database of sensitive information but wants to release statistics about how a {\em sensitive} attribute (say, disease) in the database relates to some {\em nonsensitive} attributes (e.g., postal code, age, gender, etc). This setting is widely considered in the literature, partly since it arises with medical data. Specifically, we show one can mount linear reconstruction attacks based on any release that gives:
\begin{enumerate}
\item \label{onea} the fraction of records that satisfy a given non-degenerate boolean function. Such releases include contingency tables (previously studied by Kasiviswanathan \emph{et al.}\ \cite{KRSU10}) as well as more complex outputs like the error rate of classifiers such as decision trees;

\item \label{twob} any one of a large class of $M$-estimators (that is, the output of empirical risk minimization algorithms), including the standard estimators for linear and logistic regression.
\end{enumerate}
We make two contributions: first, we show how these types of releases can be transformed into a {\em linear} format, making them amenable to existing polynomial-time reconstruction algorithms. This is already perhaps surprising, since many of the above releases (like $M$-estimators) are obtained by solving highly nonlinear formulations. % Once the
% reconstruction attack is formulated as a noisy linear system, then several standard polynomial-time methods (say minimizing $\ell_2$ or $\ell_1$ error) can be used to try to estimate the sensitive data.

Second, we show how to analyze the resulting attacks under various distributional assumptions on the data. Specifically, we consider a setting in which the same statistic (either~\ref{onea} or~\ref{twob} above) is released about how the sensitive attribute relates to all subsets of size $k$ (out of a total of $d$) nonsensitive boolean attributes. % Let $n$ be the number of records in the database. For non-degenerate boolean functions over $k+1$ variables, we show that these attacks can tolerate any noise (error in the reported statistic) that is $o(1/\sqrt{n})$ per entry and requires only $d \approx n^{1/k}$ nonsensitive attributes. For linear and logistic regression estimators, we show that these attacks can tolerate any noise that is $o(1/\sqrt{n})$ per entry and requires only $d \approx n$ nonsensitive attributes.  The analysis also extends to general $M$-estimators associated with loss functions that have Lipschitz continuous gradients. Finally, in all cases we show that the $\ell_1$ error minimization attack can also handle a constant fraction of entries with arbitrarily high noise.
\end{abstract}
%\thispagestyle{empty}
%\newpage
%\setcounter{page}{1}
%\pagenumbering{arabic}

\section{Introduction}
The goal of \emph{private data analysis} is to provide global, statistical properties of a database of sensitive information while protecting the privacy of the individuals whose records the database contains. There is a vast body of work on this problem in statistics and computer science. 

Until a few years ago, most schemes proposed in the literature lacked rigor: typically, the schemes had either no formal privacy guarantees or ensured security only against a specific suite of attacks. The seminal results of Dinur and Nissim~\cite{DiNi03} and Dinur, Dwork and Nissim \cite{DiDwNi} initiated a rigorous study of the tradeoff between privacy and utility. The notion of \emph{differential privacy}~(Dwork, McSherry, Nissim and Smith~\cite{DMNS06}, Dwork~\cite{Dwork06}) that emerged from this line of work provides rigorous guarantees even in the presence of a malicious adversary with access to arbitrary side information.  Differential privacy requires,
roughly, that any single individual's data have little effect on the outcome of the analysis. Recently, many techniques have been developed for designing differentially private algorithms. A typical objective is to release as accurate an approximation as possible to some high-dimensional function $f$ evaluated on the database~$D$.

A complementary line of work seeks to establish lower bounds on how much distortion is necessary for particular functions $f$. Some of these bounds apply only to specific notions of privacy (\eg lower bounds for differential privacy~\cite{DMNS06,GRS09,HT10,MMPRTV10,De11}). A second class of  bounds  rules out {\em any} reasonable notion of privacy by giving algorithms to reconstruct almost all of the  data $D$ given sufficiently accurate approximations to $f(D)$~\cite{DiNi03,DMT07,DY08,KRSU10,De11}. We refer to the latter results as \emph{reconstruction attacks}. 

We consider reconstruction attacks against \emph{attribute privacy}: consider a curator who manages a database of sensitive information but wants to release statistics about how a {\em sensitive} attribute (say, disease) in the database relates to some {\em nonsensitive} attributes (e.g., postal code, age, gender, etc). This setting is widely considered in the applied data privacy literature, partly since it arises with medical and retail data.\anote{substantiate this.}

For concreteness, consider a database $D$ that contains, for each individual $i$, a sensitive attribute $s_i\in\{0,1\}$ as well as some other information $\D_i \in \R^d$ which is assumed to be known to the attacker. The $i$th record is thus $(\D_i,s_i)$. We denote the entire database $D=(\D|\s)$ where $\D\in \R^{n\times d}$, $\s\in\bit{n}$, and $|$ denote concatenation. Given some released information $\y$, the attacker constructs an estimate $\hat\s$ that she hopes is close to $\s$.  We measure the attack's success in terms of the Hamming distance $d_H(\s,\hat\s)$. A scheme is \emph{not} attribute private if an attacker can consistently get an estimate that is within distance $o(n)$. Formally:

\begin{definition}[Failure of Attribute Privacy]\!\footnote{This definition generalizes \emph{blatant non-privacy} (Dinur and Nissim~\cite{DiNi03}) and first appeared in \cite{KRSU10}. The order of the qualifiers here has been changed, correcting an error pointed out by Graham Cormode.} \label{def:apy} 
A (randomized) mechanism $\mathcal{M} \,:\, \R^{n \times d+1} \to \R^a$ is said to allow $(\alpha,\beta)$ attribute reconstruction if there exists a setting of the nonsensitive attributes $\D \in \R^{n \times d}$ and an algorithm (adversary) $\AAA \::\; \R^{n \times d} \times \R^a \to \R^n$ such that for every $\s \in \{0,1\}^n$,
\begin{align*}
\Pr_{\y \gets \mathcal{M}((\D|\s))} [\AAA(\D,\y) = \hat{\s} \,: \, d_H(\s,\hat{\s}) \leq \alpha  ] \geq 1 - \beta.
\end{align*}
\end{definition}

Asymptotically, we say that a mechanism is \emph{attribute nonprivate} if there is an infinite sequence of $n$ for which $\mathcal M$ allows $(o(1), o(1))$-reconstruction. Here $d=d(n)$ is a function of $n$.  We say the attack $\AAA$ is \emph{efficient} if it runs in time $\mbox{poly}(n,d)$. 

Instead of simply showing that a setting of $\D$ exists, we will normally aim to show that reconstruction is possible with high probability when $\D$ is chosen from one of a class of natural distributions.

\paragraph{Linear Reconstruction Attacks.}
In this paper, we consider the power of \emph{linear} reconstruction attacks. Given the released information $\y$, the attacker constructs a system of approximate linear equalities, namely a matrix $A$ and vector $\z$ such that $A \s \approx \z$ and attempts to solve for $\s$. A typical algorithmic approach is to find $\hat{\s}$ which minimizes some norm ($\ell_2$ or $\ell_1$) of the error $(A\hat\s - \z)$. Minimizing the $\ell_2$ error is known as {\em least squares decoding}  and minimizing the $\ell_1$ error is known as  {\em LP decoding}.  One sometimes also considers algorithms that exhaustively search over all $2^n$ possible choices for $\s$ (as in~\cite{DiNi03,MuthuN12}).
%The attack {\em succeeds} if with high probability it recovers a large fraction of $\s$.

Such attacks were first considered in the context of data privacy by Dinur and Nissim~\cite{DiNi03}. They showed that any mechanism which answers (or allows the user to compute) $\Omega(n\log n)$ \emph{random} inner product queries with $\bit{}$  vectors on a database $\s\in \{0,1\}^n$ with $o(\sqrt{n})$ noise per query is not private. That is, they assume that the mechanism releases $\y = A\s +\mathbf{e}$, where $A$ is a random matrix in $\bit{\Omega(n\log n) \times n}$ and $\mathbf{e}$ is a noise vector with $\|\mathbf{e}\|_\infty =o(\sqrt{n})$.
% \footnote{This result of~\cite{DiNi03} is a special case of the result on releasing conjunction tables by Kasiviswanathan~\emph{et al.}\ \cite{KRSU10} which in turn is a special case of the result on releasing evaluations of non-degenerate boolean functions discussed in this paper. Note that the $o(1/\sqrt{n})$ noise bound is for releasing the {\em fraction} of records satisfying a particular non-degenerate boolean function. For releasing the actual number of records (like in the case of inner-product), we get a $o(\sqrt{n})$ noise bound.}.
Their attack was subsequently extended to use a linear number of queries~\cite{DMT07}, allow a small fraction of answers to be arbitrarily distorted~\cite{DMT07}, and run significantly more quickly~\cite{DY08}.  %However, the types of queries made by these attacks appeared quite artificial (either uniformly random or drawn from an algebraic code).

In their simplest form, such inner product queries require the adversary to be able to ``name rows'', that is, specify a coefficient for each component of the vector $\s$. Thus, the lower bound does not seem to apply to any functionality that is symmetric in the rows of the database (such as, for example, ``counting queries'').\!\footnote{It was pointed out in~\cite{DiDwNi} that in databases with more than one entry per row, random inner product queries on the sensitive attribute vector $\s$ can be simulated via hashing: for example, the adversary could ask for the sum the function $H(\D_i)\cdot s_i$ over the whole database, where $H:\bit{d-1}\to\bit{}$ is an appropriate hash function. This is a symmetric statistic, but it is unlikely to come up in a typical statistical publication.}

\paragraph{Natural Queries.}
This paper focuses on linear attacks mounted based on the release of natural, symmetric statistics. A first attack along these lines appeared in a previous work of ours (together with J. Ullman)~\cite{KRSU10} in which we analyzed the release of marginal
tables (also called contingency tables). Specifically, in~\cite{KRSU10}, we showed that any mechanism which releases the marginal distributions of all subsets of $k+1$\!\footnote{For asymptotic statements, $k$ is considered constant in this paper, as in previous works~\cite{KRSU10,De11}.} attributes with $o(\sqrt{n})$ noise
% \footnote{Again the noise bounds of~\cite{KRSU10,De11} need to be
% scaled by down by $n$ to get them in the same scale as this paper.}
per entry is attribute non-private if $d =
\tilde{\Omega}(n^{1/k})$.\!\footnote{The $\tilde{\Omega}$ notation hides polylogarithmic factors.} These noise bounds were improved in~\cite{De11}, which presented an attack that can tolerate a constant fraction of entries with arbitrarily high noise, as long as the remaining positions have $o(\sqrt{n})$ noise. We generalize both these results in this paper.

Recently, linear attacks were also considered based on range query releases~\cite{MuthuN12} (which, again, are natural, linear queries).

\paragraph{Our Results.} We greatly expand the applicability of linear attacks in ``natural'' settings. Specifically, we show one can mount linear reconstruction attacks based on any release that gives:

\begin{enumerate}
\item \label{one} the fraction of records that satisfy a given {\em non-degenerate} boolean function (a boolean function over $p$ variables is non-degenerate if its multilinear representation has degree exactly $p$). Such releases include contingency tables as well as more complex outputs like the error rate of certain classifiers such as decision trees; or

\item \label{two} the $M$-estimator associated with a differentiable loss function.  $M$-estimators are a broad class of estimators which are obtained by minimizing sums of functions of the records (they are also called \emph{empirical risk minimization} estimators). $M$-estimators include the standard estimators for linear and logistic regression (both these estimators are associated with differentiable loss functions). See Section~\ref{sec:mest} for definitions.
\end{enumerate}

Our contributions are two-fold. First, we show how these types of releases can be transformed into a (noisy) linear release problem, making them amenable to linear reconstruction attacks.  This is already perhaps surprising, since many of the above statistics (like $M$-estimators) are obtained by solving highly nonlinear formulations. After performing this transformation, we can apply polynomial-time methods (like least squares or LP decoding) on this linear release problem to estimate the sensitive data.
%Since we are considering all subsets of size $k$ out of $d$, there are $\approx d^k$ entries released for each statistic.

Second, we show how to analyze these attacks under various distributional assumptions on the data. This gives lower bounds on the noise needed to release these statistics attribute privately. Specifically, we consider a setting in which the same statistic (either~\ref{one} or~\ref{two} above) is released about how the sensitive attribute relates to all subsets of (constant) size $k$ (out of a total of $d$) nonsensitive boolean attributes.  For a subset $J \subseteq [d]$ of size $k$, let $\D|_J$ denote the submatrix of $\D$ consisting of the columns in $J$.

The $\binom d k$ entries for a statistic are obtained by evaluating the same statistic on $(\D|_J|\s)$  for all sets $J$ of size $k$.  Specifically:
\begin{CompactItemize}
\item 
Consider a mechanism which releases, for every set $J$ of size $k$, the fraction of records (rows) in $(\D|_J|\s)$ satisfying a non-degenerate boolean function over $k+1$ variables. We show that if the mechanism adds $o(1/\sqrt{n})$ noise per entry and if $d=\tilde{\Omega}(n^{1/k})$, then it is attribute non-private. 
\item   Consider a mechanism which releases, for every set $J$ of size $k$, a particular $M$-estimator evaluated over  $(\D|_J|\s)$. We show that if the mechanism adds $o(1/(\lambda\sqrt{n}))$ noise per entry and if $d=\Omega(n)$, then it is attribute
non-private. Here, $\lambda$ is the Lipschitz constant of the loss function gradient. The loss function also needs to satisfy a mild variance condition. For the case of linear and logistic regression estimators, $\lambda=\Theta(1)$ for bounded data, and so the noise bound is $o(1/\sqrt{n})$. 
\end{CompactItemize} 
The statements above are based on the least squares attack. For most settings, we show the LP decoding attack can also handle a constant fraction of entries with arbitrarily high noise (the exception is the setting of general $M$-estimators).

\paragraph{Techniques for Deriving the Attacks.} Casting the releases as a system of linear equations requires two simple insights which we hope will be broadly useful. First, we note that when $\s=(s_1 \etc s_n)$ is boolean, then any release which allows us to derive an equation which involves a sum over database records can in fact be made linear in $\s$. Specifically, suppose we know that $\sum_i g_i(s_i) = t$, where $t$ is a real number and $g_i$ is an arbitrary real-valued function that could depend on the index $i$, the public record $\D_i$, and any released information. We can rewrite $g_i(s_i)$ as $g_i(0) + s_i(g_i(1)-g_i(0))$; the constraint $\sum_i g_i(s_i) = t$ can then be written as $\sum_i s_i\cdot (g_i(1)-g_i(0)) = t-\sum_i g_i(0)$, which is affine in $\s$. This allows us to derive linear constraints from a variety of not-obviously-linear releases; for example, it allows us to get linear attacks based on the error rate of a given binary classifier (see Section~\ref{sec:multilinear}).

The second observation is that for many \emph{nonlinear} optimization problems, every optimal solution must satisfy constraints that are, in fact, sums over data records. For example, for $M$-estimators associated with differentiable loss functions, the gradient at the solution $\hat \theta$ must equal $0$, leading to an equation of the form $\sum_i \partial\, \ell(\hat\theta;(\D_i,s_i)) =0$. This can be made linear in $\s$ using the first technique. We bound the effect of any noise added to the entries of $M$-estimator ($\hat \theta$) via the Lipschitz properties of the gradient of the loss function~$\ell$. 

\paragraph{Techniques for Analyzing the Attacks.} The techniques just mentioned give rise to a variety of linear reconstruction attacks, depending on the type of released statistics. We can provide theoretical guarantees on the performance of these attacks in some settings, for example when the same statistic is released about many subsets of the data (\eg all sets of a given size $k$) and when the data records themselves are drawn i.i.d.\ from some underlying distribution. The main technique here is to analyze the geometry of the constraint matrix $A$ that arises in the attack. For the case of non-degenerate boolean functions, we do so by relating the constraint matrix to a row-wise product of a matrix with i.i.d.\ entries (referred to as a random row product matrix, see Section~\ref{sec:randconj}), which was recently analyzed by Rudelson~\cite{R11} (see also~\cite{KRSU10}). The results of~\cite{R11} showed that the least singular value of a random row product matrix is asymptotically the same as that of a matrix of same dimensions with i.i.d.\ entries, and a random row product matrix is a Euclidean section. Our results show that a much broader class of matrices with correlated rows satisfy these properties. 
% It is an interesting open question to analyze the attacks in more general settings, and to understand (experimentally) how well the attacks perform on ``actual'' releases.

\paragraph{Organization.} In Section~\ref{sec:prelim}, we introduce some notation and review the least squares and LP decoding techniques for solving noisy linear systems. In Section~\ref{sec:multilinear}, we present our results on evaluating non-degenerate boolean functions. As mentioned earlier, we first reduce the release problem to a linear reconstruction problem (Section~\ref{sec:lower}), and then the attacks works by using either least squares or LP decoding techniques. The analysis requires analyzing spectral and geometric properties of the constraint matrix that arises in these attack which we do in Section~\ref{sec:spect}. In Section~\ref{sec:mest}, we present our results on releasing $M$-estimators associated with differentiable loss functions. For clarity, we first discuss the attacks for the special cases of linear and logistic regression estimators (in Section~\ref{sec:linlog}), and then discuss the attacks for the general case (in Section~\ref{sec:genfn}).

\section{Preliminaries} \label{sec:prelim}
\mypar{Notation.} We use $[n]$ to denote the set $\{1 \etc n\}$. $d_H(\cdot,\cdot)$ measures the Hamming distance. Vectors used in the paper are by default column vectors and are denoted by boldface letters. For a vector $\v$, $\v^\top$ denotes its transpose, $\|\v\|$ denotes its Euclidean norm, $\|\v\|_1$ denotes its $\ell_1$ norm, and $\|\v\|_{\infty}$ denotes its $\ell_{\infty}$ norm. For two vectors $\v_1$ and $\v_2$, $\langle \v_1, \v_2 \rangle$ denotes the inner product of $\v_1$ and $\v_2$. We use $(a)_n$ to denote a vector of length of $n$ with all entries equal to $a$. For a matrix $M$, $\|M\|$ denotes the operator norm and $M_i$ denotes the $i$th row of $M$. Random matrices are denoted by boldface capitalized letters. We use $\mathrm{\it diag}(a_1 \etc a_n)$ to denote an $n \times n$ diagonal matrix with entries $a_1 \etc a_n$ along the main diagonal. The notation $\mbox{vert}(\cdot\etc\cdot)$ denotes vertical concatenation of the argument matrices.

Let $M$ be an $N \times n$ real matrix with $N \geq n$. The singular values $\sigma_j(M)$ are the eigenvalues of $\sqrt{M^\top M}$ arranged in non-increasing order. Of particular importance in this paper is the least singular value $\sigma_n(M) = \inf_{\z: \|\z\| =1 } \|M\z\|$. The unit sphere in $n$ dimensions centered at origin is denoted by $S^{n-1} = \{\z \,:\, \|\z\| =1 \}$. 

Our analysis uses random matrices, and we add a subscript of $r$ to differentiate a random matrix from a non-random matrix. As mentioned earlier,  $k$ is a constant in this paper and we often omit dependence on $k$ in our results. %In Appendix~\ref{app:boolprelim}, we discuss preliminaries about boolean function which are useful for our paper.   

\subsection{Background on Noisy Linear Systems.}\label{sec:actrecon}
Noisy linear systems arise in a wide variety of statistical and signal-processing contexts. Suppose we are given a matrix $A$ and vector $\z$ such that $\z = A\s+\mathbf{e}$, where $\mathbf{e}$ is assumed to be ``small'' (in a sense defined below). A natural approach to estimating $\s$ is to
output $\hat \s = \mbox{argmin}_{\s}\, \|A\s-\z\|_p$ for some $p\geq 1$. We will consider $p=1$ and $2$; we summarize the assumptions and guarantees for each method below. When it is known that $\s\in\bit{n}$, the attacker can then round the entries of $\hat \s$ to the nearer of $\bit{}$ to improve the estimate.

In the sequel, we call a vector $\mathbf{z} \in \R^m$ \emph{$(a,b)$-small} if at least $1-a$ fraction of its entries have magnitude less than $b$. In other words, for some set $S$, $|S| \geq (1-a) \cdot m$, it is the case that $|z_i| \leq b$ for all $i \in S$. 

\mypar{$\ell_2$ error minimization (``least squares'').} 
Widely used in regression problems, the least squares method guarantees a good approximation to $\s$ when the Euclidean norm $\|\mathbf{e}\|$ is small and $A$ has no small eigenvalues. It was first used in data privacy by Dwork and Yekhanin~\cite{DY08}. For completeness, we present the entire analysis of this attack (in a general setting) here. 

Let $A = P \Sigma Q^\top$ be the singular value decomposition of $A$. Here, $P$ is an orthogonal $m\times m$ matrix,  $\Sigma$ is a  diagonal $m\times n$ matrix, and $Q$ is an orthogonal $n \times n$ matrix. Let $\mathbf{0}_{n \times (m-n)}$ be an $n \times (m-n)$ matrix with all entries zero. Define 
$$\Gamma_{\mathrm{inv}} = (\mathrm{\it diag}(\sigma_{n}(A)^{-1} \etc \sigma_1(A)^{-1})|\mathbf{0}_{n \times (m-n)}).$$ 
The dimension of $\Gamma_{\mathrm{inv}}$ is $n \times m$. Define $A_\mathrm{inv}=Q\Gamma_{\mathrm{inv}}P^\top.$ 

Given $\y$, the adversary uses $A_\mathrm{inv}$ to construct $\hat{\s} = (\hat{s}_1 \etc \hat{s}_n)$ as follows: 
\begin{align*}
\hat{s}_i = \left\{ \begin{array}{rl}
0 &\mbox{ if $i$th entry in $A_\mathrm{inv}\z< 1/2$}, \\
1 &\mbox{ otherwise.}
\end{array} \right.
\end{align*}
In other words, $\hat{\s}$ is obtained by rounding $A_\mathrm{inv}\z$ to closest of $0,1$.

Now  the claim is that $\hat{\s}$ is a good reconstruction of $\s$. The idea behind the analysis is that $A_\mathrm{inv}\z= \s + A_\mathrm{inv}\mathbf{e}$. 
Now (as $P$ and $Q$ are orthogonal matrices, they don't affect the norms),
$$\|A_\mathrm{inv}\mathbf{e}\| =\|Q \Gamma_{\mathrm{inv}}P^\top \mathbf{e}\|  \leq \|\Gamma_{\mathrm{inv}}\| \|P^\top \mathbf{e}\| =  \frac{\|\mathbf{e}\|}{\sigma_n(A)}.$$

Let us assume that $\sigma_n(A) = \sigma$.  If (the absolute value of) all the entries in $\mathbf{e}$ are less than $\beta$ then $\|\mathbf{e}\| =  \beta \sqrt{m}$, and therefore $\|A_\mathrm{inv}\mathbf{e} \| = (\beta \sqrt{m})/\sigma$. In particular, this implies that $A_\mathrm{inv}\mathbf{e}$ cannot have $(4 m \beta^2)/\sigma^2$ entries with absolute value above $1/2$, and therefore the Hamming distance between $\hat{\s}$ and $\s$ is $O(4 m \beta^2)/\sigma^2$ (as the adversary only fails to recover those entries of $\s$ whose corresponding $A_\mathrm{inv}\mathbf{e}$ entries are greater than $1/2$). The time complexity of the attack is dominated by the cost of computing the singular value decomposition of $A$ which takes $O(mn^2)$ time.\!\footnote{SVD decomposition of $N_1 \times N_2$ sized matrix can be done in $O(N_1 N_2^2)$ time.} 
\begin{theorem} \label{thm:KRSU}
Let $A \,:\, \R^{n} \to \R^m$ be a full rank linear map with $(m > n)$ such that the least singular value of $A$ is $\sigma$. Then if $\mathbf{e}$ is $(0,\beta)$ small (that is, if $\|\mathbf{e}\|_\infty \leq \beta$), the vector $\mbox{argmin}_{\s}\, \|A\s-\z\|$, rounded to $\bit{n}$, satisfies
$d_H(\s,\hat{\s}) \leq (4 m \beta^2)/\sigma^2$. In particular, if $\sigma = \Omega(\sqrt{m})$ and $\beta=o(\sqrt{n})$, then $\hat\s$ agrees with a $1-o(1)$ fraction of
$\s$. The attack runs in $O(m n^2)$ time.
\end{theorem}

\mypar{$\ell_1$ error minimization (``LP decoding'').} In the context of privacy, the ``LP decoding'' approach was first used by Dwork \emph{et al.}\ \cite{DMT07}. (The name stems from the fact that the minimization problem can be cast as a linear program.) The LP attack is slower than the least squares attack but can handle considerably more complex error patterns at the cost of a stronger assumption on $A$.  Recently, De~\cite{De11} gave a simple analysis of this attack based on the geometry of the operator $A$. We need the following definition of Euclidean section.
 
\begin{definition} \label{def:euclidean}
A linear operator $A \,: \, \R^n \to \R^m$ is said to be a $\alpha$-Euclidean section if for all $\s$ in $\R^n$, 
$$\sqrt{m} \|A\s\| \geq \|A\s\|_1 \geq \alpha \sqrt{m} \|A\s\|\,.$$ 
Note that by Cauchy-Schwarz, the first inequality, $\sqrt{m} \|A\s\| \geq \|A\s\|_1$, always holds. We remark that when we say $A$ is Euclidean, we simply mean that
there is some constant $\alpha > 0$ such that $A$ is $\alpha$-Euclidean.
\end{definition}

The following theorem gives a sufficient condition under which LP decoding gives a good approximation to $\s$. The time bound here was derived from the LP algorithm of Vaidya~\cite{Vaidya}, which uses $O(((N_1+N_2)N_2^2 + (N_1+N_2)^{1.5} N_2)N_3)$ arithmetic operations where $N_1$ is the number of constraints, $N_2$ is the number of variables, and $N_3$ is a bound on the number of bits used to describe the entries of $A$.  In our setting, the LP has $n$ variables, $m$ constraints, and $N_3$ could be upper bounded by $mn$, and therefore the LP could be solved in $O(m^2n^3 + m^{2.5}n^2)$ time.

\begin{theorem}[From~\cite{De11}] \label{thm:d11} Let $A \,:\, \R^n \to \R^m$ be a full rank linear map ($m > n$) such that the least singular value of $A$ is $\sigma$. Further, let $A$ be a $\alpha$-Euclidean section. Then there exists a $\gamma=\gamma(\alpha)$ such that if $\mathbf{e}$ is $(\gamma,\beta)$ small, then any solution $\hat{\s} = \mbox{argmin}_{\s}\, \|A\s-\z\|_1$, rounded to $\bit{n}$, satisfies $d_H(\s,\hat{\s}) \leq O(\beta \sqrt{mn}/\sigma)$ where the constant inside the $O(\cdot)$ notation depends on $\alpha$. In particular, if $\sigma = \Omega(\sqrt{m})$ and $\beta=o(\sqrt{n})$, then this attack recovers $1-o(1)$ fraction of $\s$. The attack runs in $O(m^2n^3 + m^{2.5}n^2)$ time.
\end{theorem}

\section{Releasing Evaluations of Non-Degenerate Boolean Functions} \label{sec:multilinear}
In this section, we analyze privacy lower bounds for releasing evaluations of non-degenerate boolean functions. We use the following standard definition of representing polynomial (see Appendix~\ref{app:boolprelim} for a background about representing boolean functions as multilinear polynomials).
\begin{definition}
A polynomial $P^{(f)}$ over the reals represents a function $f$ over $\{0,1\}^{k+1}$ if $f(x_1 \etc x_{k+1})  = P^{(f)}(x_1 \etc x_{k+1})$ for all $(x_1 \etc x_{k+1}) \in \{0,1\}^{k+1}$.
\end{definition}
%Let $\mathcal{F}$ be the set of all boolean functions from $\{0,1\}^{k+1} \to \{0,1\}$. We concentrate on a subset of functions from $\mathcal{F}$, which we call as {\em non-degenerate}. 
\begin{definition}\label{def:deg1}
A function $f: \{0,1\}^{k+1} \to \{0,1\}$ is non-degenerate iff it can be written as a multilinear polynomial of degree $k+1$.
\end{definition}
If $f$ is non-degenerate, then it depends on all of its $k+1$ variables. Note that non-degenerate functions constitute  a large class of functions.\!\footnote{A simple counting argument shows that among the $2^{2^{k+1}}$ boolean functions over $k+1$ variables, $2^{2^{k+1}} - \binom{2^{k+1}}{2^{k}}$ are non-degenerate.} For example, it includes widely used boolean functions like AND, OR, XOR, MAJORITY, and depth $k+1$ decision trees \cite{boolfunc}. %Since most of the interesting functions that can be represented as multilinear polynomials are boolean (like the ones noted above), we show a lower bound for the boolean database. That is $D \in (\{0,1\})_{n \times d+1}$. 
%Missing details from this section are collected in Appendix~\ref{app:multilinear}.
 
\mypar{Problem Statement.} Let $f : \{0,1\}^{k+1} \to \{0,1\}$ represent the function that we want to evaluate on a database $D$. Let $D=(\D|\s) \in (\{0,1\})_{n \times (d+1)}$, where $\D \in (\{0,1\})_{n \times d}$ and $\s \in \{0,1\}^n$. Let $\D=(\d_{i,j})$, i.e., $\d_{i,j}$ denotes the $(i,j)$th entry in $\D$ (with $1 \leq i \leq n$ and $1 \leq j \leq d$). Let
\begin{align*}
& J =(j_1 \etc j_{k})  \in \{1 \etc d \}^{k} \\ &(\mbox{where }  \{1 \etc d\}^k = \underbrace{\{1 \etc d\} \times \dots \times \{1 \etc d\}}_{k \mbox{ times}}). 
\end{align*}
%$J =(j_1 \etc j_{k})  \in \{1 \etc d \}^{k} \;\;(\mbox{where } \{1 \etc d\}^k = \{1 \etc d\} \times \dots \times \{1 \etc d\} \mbox{ repeated $k$ times}). $
Note that $J$ allows repeated entries.\!\footnote{We allow repeated entries for convenience of notation. Our results also hold if we use the more natural $J \subseteq [d]$, $|J| = k$.} Let $D|_J$ be the submatrix of $D$ restricted to columns indexed by $J$. For a fixed $J$, define $F(D|_J)$~as
$$F(D|_J) = \sum_{i=1}^n f(\d_{i,j_1} \etc \d_{i,j_k},s_i),\;\;\; J = (j_1 \etc j_{k}).$$
Note that $F(D|_J)$ is an integer between $0$ to $n$.  Let $\Sigma_f(D)$ be the vector obtained by computing $F$ on all different $D|_J$'s: $\Sigma_f(D) = ( F(D|_J) ) \mbox{ where } J \in \{1 \etc d\}^{k}$. Note that $\Sigma_f(D)$ is a vector of length $d^{k}$. The goal is to understand how much noise is needed to attribute privately release $\Sigma_f(D)$ (or $\Sigma_f(D)/n$) when $f$ is non-degenerate. 

\mypar{Our Results.} We prove the following results using the $\ell_2$ and $\ell_1$ error minimization attacks outlined in Section~\ref{sec:actrecon}.
\begin{theorem}[Informal Statements] \label{thm:multlin} 
Let $f:\{0,1\}^{k+1} \to \{0,1\}$ be a non-degenerate boolean function. Then 
\begin{CompactEnumerate}
\item \label{part1} any mechanism which for every database $D \in (\{0,1\})_{n \times (d+1)}$ with $n \ll  d^{k}$ releases $\Sigma_f(D)$ by adding $o(\sqrt{n})$ (or releases $\Sigma_f(D)/n$ by adding $o(1/\sqrt{n})$) noise to each entry is attribute non-private. The attack that achieves this non-privacy violation runs in  $O(d^{k} n^2)$ time.
\item \label{part2} there exists a constant $\gamma > 0$ such that any mechanism which for every database $D \in (\{0,1\})_{n \times (d+1)}$ with $n \ll d^{k}$ releases $\Sigma_f(D)$ by adding $o(\sqrt{n})$ (or releases $\Sigma_f(D)/n$ by adding $o(1/\sqrt{n})$) noise to at most $1-\gamma$ fraction of the entries is attribute non-private. The attack that achieves this non-privacy violation runs in $O(d^{2k} n^3 + d^{2.5 k}n^2)$~time.
\end{CompactEnumerate}
\end{theorem} 
For convenience of notation, in this section, we work mostly with the transpose of $\D$. Let $\T = \D^\top$. So $\T$ is a $d \times n$ matrix. 
\subsection{Reducing to a Linear Reconstruction Problem.} \label{sec:lower}
In this section, we reduce the problem of releasing $\Sigma_f(D)$ for a database $D$ into a linear reconstruction problem. First, we define a simple decomposition of boolean functions. Consider a non-degenerate boolean function $f \,: \{0,1\}^{k+1} \to \{0,1\}$. Now  there exists two function $f_0 \,:\, \{0,1\}^{k} \to \{0,1\}$ and $f_1 \,:\, \{0,1\}^{k} \to \{0,1\}$ such that
\[ f(\d_1 \etc \d_{k+1}) = f_0(\d_1 \etc \d_{k})(1-\d_{k+1}) + f_1(\d_1 \etc \d_{k}) \d_{k+1} \;\;\;\; \forall (\d_1 \etc \d_{k+1}) \in \{0,1\}^{k+1}. \]
This can be re-expressed as
\[ f(\d_1 \etc \d_{k+1})  = f_0(\d_1 \etc \d_{k}) + (f_1(\d_1 \etc \d_{k}) -f_0(\d_1 \etc \d_{k})) \d_{k+1}. \]
Define $f_{2}(\d_1 \etc \d_{k}) = f_1(\d_1 \etc \d_{k}) -f_0(\d_1 \etc \d_{k})$.  Therefore, 
\begin{equation}  \label{eqn:fdecomp} f(\d_1 \etc \d_{k+1})  = f_0(\d_1 \etc \d_{k}) + f_{2}(\d_1 \etc \d_{k}) \d_{k+1}. \end{equation}
Note that $f_{2}$ is a function from $\{0,1\}^{k} \to \{-1,0,1\}$. Since both $f_0$ and $f_1$ are both boolean functions and can be represented as multilinear polynomials over the variables $\d_1 \etc \d_{k}$, therefore $f_{2}$ also could be represented as a multilinear polynomial over the variables $\d_1 \etc \d_{k}$. Since $f$ is represented by a multilinear polynomial of degree $k+1$, therefore, the multilinear polynomial representing $f_{2}$ has degree $k$ (if it has any lower degree, then $f$ could be represented as multilinear polynomial of degree strictly less than $k+1$, which is a contradiction). %Additionally, $1/c_{f_{2}}(k) \geq 2^{-k}$. 
To aid our construction, we need to define a particular function of matrices. 
\begin{definition}[Row Function Matrix] \label{def:funcmat}
Let $h$ be a function from $\{0,1\}^k \to \{-1,0,1\}$. Let $\T_{(1)}=(\d^{(1)}_{i,j}),\T_{(2)}=(\d^{(2)}_{i,j}) \etc \T_{(k)}=(\d^{(k)}_{i,j})$ be $k$ matrices with $\{0,1\}$ entries and dimensions $d \times n$. Define a row function matrix (of dimension $d^k \times n$) $\Pi_h(\T_{(1)} \etc \T_{(k)})$ as follows. Any row of this matrix will correspond to a sequence 
$$J = (j_1, j_2 \etc j_k) \in \{1 \etc d \}^k$$ 
of $k$ numbers, so the entries of $\Pi_h(\T_{(1)} \etc \T_{(k)})$ will be denoted\footnote{The definition assumes a certain order of the rows of the matrix $\Pi_h(\T_{(1)} \etc \T_{(k)})$. This order, however, is not important, for our analysis. Note that changing the relative positions of rows of a matrix doesn't affect its eigenvalues and singular values.} by $\pi_{J,a}$, where $a \in \{1 \etc n\}$.  For $J = (j_1, j_2 \etc j_k)$ the entries of the matrix $\Pi_h(\T_{(1)} \etc \T_{(k)})$ will be defined by the relation
\[ \pi_{J,a}=h(\d^{(1)}_{j_1,a}, \d^{(2)}_{j_2,a} \etc \d^{(k)}_{j_k,a}). \] 
\end{definition}
The row product matrices from~\cite{KRSU10} (see Definition~\ref{def:entrywise}) is a particular example of this construction where the function $h(\d_1, \d_2 \etc \d_k)= \d_1\cdot \d_2 \cdot \ldots \cdot \d_k,$, which implies that $\Pi_h(\T_{(1)} \etc \T_{(k)}) = \T_{(1)} \odot \dots \odot \T_{(k)}$, where $\odot$ is the row-product operator from Definition~\ref{def:entrywise}. 

Let $D=(\D|\s)$ be a database, and let $\T = \D^\top$. Let $\D =(\d_{i,j})$. Consider any fixed $J = (j_1 \etc j_{k}) \in \{1 \etc d\}^{k}$. Now for this $J$, there exists an entry in $\Sigma_f(D)$ equaling $\sum_{i=1}^n f(\d_{i,j_1} \etc \d_{i,j_{k}},s_i)$. Now consider the matrices $\Pi_{f_{2}}(\T \etc \T)$ and $\Pi_{f_0}(\T \etc \T)$. Consider the rows in $\Pi_{f_0}(\T \etc \T)$ and $\Pi_{f_{2}}(\T \etc \T)$ corresponding to this above $J$. Let this be the $l$th row in these matrices. Then the $l$th row of the matrix $\Pi_{f_0}(\T \etc \T)$ has $n$ entries equaling $f_0(\d_{i,j_1} \etc \d_{i,j_{k}})$ and the $l$th row of the matrix $\Pi_{f_{2}}(\T \etc \T)$ has $n$ entries equaling $f_{2}(\d_{i,j_1} \etc \d_{i,j_{k}})$ for $i=1 \etc n$. Since 
\begin{equation*}f(\d_{i,j_1} \etc \d_{i,j_{k}},s_i) = f_0(\d_{i,j_1} \etc \d_{i,j_{k}}) \\ +f_{2}(\d_{i,j_1} \etc \d_{i,j_{k}})s_i,\end{equation*} it follows that 
\begin{align*}  \sum_{i=1}^n f(\d_{i,j_1} \etc \d_{i,j_{k}},s_i) & = \sum_{i=1}^n f_0(\d_{i,j_1} \etc \d_{i,j_{k}})+f_{2}(\d_{i,j_1} \etc \d_{i,j_{k}})s_i  \\
&  = \langle \Pi_{f_0}(\T \etc \T)_l, \mathbf{1}_n  \rangle + \langle \Pi_{f_{2}}(\T \etc \T)_l, \s \rangle, \end{align*}
where $\Pi_{f_0}(\T \etc \T)_l$ and $\Pi_{f_{2}}(\T \etc \T)_l$ denote the $l$th row of matrices $\Pi_{f_0}(\T \etc \T)$ and $\Pi_{f_{2}}(\T \etc \T)$ respectively and $\mathbf{1}_n$ denotes the vector $(1)^n$.
Now define a vector $H_{f}(D)$ whose $l$th element ($1 \leq l \leq d^k$) is
\begin{equation}\label{eqn:hf}   H_{f}(D)_l  = \langle \Pi_{f_0}(\T \etc \T)_l, \mathbf{1}_n  \rangle + \langle \Pi_{f_{2}}(\T \etc \T)_l, \s  \rangle.\end{equation}
The length of vector $H_{f}(D)$ is $d^{k}$. The above arguments show that all the entries of $H_f(D)$ are contained in the vector $\Sigma_f(D)$. Since every row in these $\Sigma_f(D)$ correspond to some $J$, it also follows that all the entries of $\Sigma_f(D)$ are contained in the vector $H_f(D)$, implying the following claim.
\begin{claim} \label{claim:eq}
$\Sigma_f(D) = H_{f}(D)$.\!\footnote{Under proper ordering of both the vectors.}
\end{claim}

\mypar{Setting up the Least Squares Attack.} The privacy mechanism releases a noisy approximation to $\Sigma_f(D)$. Let $\y=(y_1 \etc y_{d^k})$ be this noisy approximation. 
The adversary tries to reconstruct an approximation of $\s$ from $\y$. Let $b_{f_i} = \langle \Pi_{f_0}(\T \etc \T)_i, \mathbf{1}_n  \rangle$, and $\mathbf{b}_f =(b_{f_1} \etc b_{f_{d^k}})$. Given $\y$, the adversary solves the following linear reconstruction problem:
\begin{align} \label{eqn:reconprob} \y = \mathbf{b}_f + \Pi_{f_{2}}(\T \etc \T)\s. \end{align}
In the setting of attribute non-privacy the adversary knows $\T$, and therefore can compute $\Pi_{f_{2}}(\T \etc \T)$ and $\Pi_{f_0}(\T \etc \T)$ (hence, $\b_f$). The goal of the adversary is to compute a large fraction of $\s$ given $\y$. The definition of iterated logarithm $(\log_{(q)})$ is given in Definition~\ref{defn:iterlog}. In the below analysis, we use a random matrix $\BT$ and the least singular value lower bound on a random row function matrix from Theorem~\ref{thm:least}. We use boldface letters to denote random matrices.
%Proposition~\ref{prop:upper} (in Appendix~\ref{app:upper}) shows that the noise bound obtained in this theorem is tight for attribute privately releasing $\Sigma_f(D)$.
\begin{theorem} [Part~\ref{part1}, Theorem~\ref{thm:multlin}] \label{thm:ls}
Let $f:\{0,1\}^{k+1} \to \{0,1\}$ be a non-degenerate boolean function and $n \leq c d^k/\log_{(q)}d$ for a constant $c$ depending only on $k,q$.  Any privacy mechanism which for every database $D \in (\{0,1\})_{n \times (d+1)}$ releases $\Sigma_f(D)$  by adding $o(\sqrt{n})$ (or releases $\Sigma_f(D)/n$ by adding $o(1/\sqrt{n})$) noise to each entry is attribute non-private. The attack that achieves this non-privacy violation runs in $O(d^{k} n^2)$~time.
\end{theorem}
\begin{proof}
Consider the least squares attack outlined in Theorem~\ref{thm:KRSU} on Equation~\eqref{eqn:reconprob}. Let $\BT$ be a random matrix of dimension $d \times n$ with independent Bernoulli entries taking values $0$ and $1$ with probability $1/2$, and let database $D=(\BT^\top|\s)$ for some $\s \in \{0,1\}^n$.  For analyzing the attack in Theorem~\ref{thm:KRSU}, we need a lower bound on the least singular value of $\Pi_{f_{2}}(\BT \etc \BT)$. The following claim follows from Theorem~\ref{thm:least}. Note that $f_{2}$ is a function over $k$ variables.
\begin{claim} \label{claim:least}
For function $f_{2}$ defined above and $n \leq c d^k/\log_{(q)}d$ (where $c$ is the constant from Theorem~\ref{thm:L1}), the matrix $\Pi_{f_{2}}(\BT \etc \BT)$ satisfies
\begin{align*} 
%& \Pr \left[ \exists \x \in S^{n-1} \ \norm{\Pi_{f_{2}}(\BT \etc \BT) \x}_1 \le  C' d^k   \right] \le  c_1 \exp \left( - c_2 d  \right),\\
& \Pr \left[ \sigma_n(\Pi_{f_{2}}(\BT \etc \BT)) \le C' \sqrt{d^k}   \right] \le  c_1 \exp \left( - c_2 d  \right). \end{align*}
\end{claim}
\begin{proof}
Apply Theorem~\ref{thm:least} with function $h = f_2$. 
\end{proof}
Claim~\ref{claim:least} shows that with exponentially high probability $\sigma_n(\Pi_{f_{2}}(\BT \etc \BT)) = \Omega(\sqrt{d^{k}})$.  Invoking Theorem~\ref{thm:least} with $m=d^k$ and $\beta=o(\sqrt{n})$ shows that with exponentially high probability the adversary fails to recover only $o(n)$ entries of $\s$. Noise bound of $o(\sqrt{n})$ for releasing $\Sigma_f(D)$ translates into a noise bound of $o(1/\sqrt{n})$ for releasing $\Sigma_f(D)/n$. 
\end{proof}

\mypar{Setting up the LP Decoding Attack.} The LP decoding attack solves a slightly different reconstruction problem than Equation~\eqref{eqn:reconprob}. The reason is because  $\Pi_{f_{2}}(\BT \etc \BT)$ is not a Euclidean section\footnote{If we take a matrix $\BT$ of dimension $d \times n$ with independent Bernoulli entries taking values $0$ and $1$ with probability $1/2$, the resulting matrix $\Pi_{f_{2}}(\BT \etc \BT)$ is not a Euclidean section. This is because the matrix $\BT$ is non-centered (expectation of each entry in the matrix is $1/2$) which makes the $\norm{\Pi_{f_{2}}(\BT \etc \BT)}$ to be $\approx d^k$ instead of $\sqrt{d^k}$. } (a property needed for applying Theorem~\ref{thm:d11}). However, we show that a related reconstruction problem has all the properties needed for the LP decoding attack. The analysis goes via matrices with $\{1,-1\}$ entries  which have the desired properties. We establish the Euclidean section property in Appendix~\ref{app:euclid}.

Let $D=(\D|\s)$ be a database, and let $\T=\D^\top$. Let $\Phi = 2\T - \mathbf{1}_{d \times n}$ where $\mathbf{1}_{d \times n}$ is a $d \times n$ matrix of all $1$'s. Define $g \,:\, \{-1,1\}^{k+1} \to \{-1,1\}$ as
\begin{equation} \label{eqn:fg}  g(\phi_1 \etc \phi_{k+1}) =2f \left(\frac{1+\phi_1}{2} \etc \frac{1+\phi_{k+1}}{2} \right) -1. \end{equation}
We can decompose $g$ as
\begin{equation*} g(\phi_1 \etc \phi_{k+1}) =  \left (\frac{1+\phi_{k+1}}{2} \right ) g_1(\phi_1 \etc \phi_{k})  + \left ( \frac{1-\phi_{k+1}}{2} \right ) g_{-1}(\phi_1 \etc \phi_{k}),\end{equation*}
where $g_1 : \{-1,1\}^{k} \to \{-1,1\}$ and $g_{-1} : \{-1,1\}^{k} \to \{-1,1\}$. Using the notation, $g_1 =g_1(\phi_1 \etc \phi_{k})$ and $g_{-1} = g_{-1}(\phi_1 \etc \phi_{k})$, we get 
$$g(\phi_1 \etc \phi_{k+1}) = (g_1 + g_{-1})/2 + \phi_{k+1}(g_1- g_{-1})/2.$$
Define  $g_{2}=g_{2}(\phi_1 \etc \phi_{k}) = (g_1 - g_{-1})/2$ and $g_3 = g_3(\phi_1 \etc \phi_{k}) = (g_1 + g_{-1})/2$. Let us denote $\d_i = (1+\phi_i)/2$. Using the decomposition of $f$ from Equation~\eqref{eqn:fdecomp},
\begin{equation*} f\left(\frac{1+\phi_1}{2} \etc \frac{1+\phi_{k+1}}{2} \right) = f(\d_1 \etc \d_{k+1}) = f_0(\d_1 \etc \d_{k}) + f_{2}(\d_1 \etc \d_{k}) \d_{k+1}.\end{equation*}
Using the notation, $f_0 = f_0(\d_1 \etc \d_{k})$ and $f_{2}=f_{2}(\d_1 \etc \d_{k})$, and substituting the decomposition of $f$ and $g$ into Equation~\eqref{eqn:fg} gives: 
\begin{equation}
\label{eqn:fgrel}
g_3 + g_{2}\phi_{k+1} = 2  (f_0 + f_{2}\d_{k+1})-1  \equiv f_0 + f_{2}\d_{k+1} = (1/2)(g_3 + g_{2}\phi_{k+1} + 1).
\end{equation}
\begin{claim}
Let $D=(\D|\s)$, $\T=\D^\top$, and $\Phi = 2\T - \mathbf{1}_{d \times n}$. Then
\begin{equation*} 
\Sigma_f(D)_i = \frac{1}{2} \left (  \langle \Pi_{g_3}(\Phi \etc \Phi)_i, \mathbf{1}_n \rangle +  \langle \Pi_{g_{2}}(\Phi \etc \Phi)_i, 2\s-\mathbf{1}_n \rangle  + n\right  ). \end{equation*}
%where $\T_i$ is the $i$th row in $\T$.
\end{claim}
\begin{proof}
Note that
from the earlier established decomposition
$$\Sigma_f(D)_i = \langle \Pi_{f_0}(\T \etc \T)_i, \mathbf{1}_n  \rangle + \langle \Pi_{f_{2}}(\T \etc \T)_i, \s \rangle.$$
The reminder of the proof follows by using Equation~\eqref{eqn:fgrel} along with the definition of row-function matrices (Definition~\ref{def:funcmat}).   
\end{proof}
%Applying the above claim, we get 
%\begin{align*}
%& \Sigma_f(D)_i = \langle \Pi_{f_0}(\T \etc \T)_i, \mathbf{1}_n  \rangle + \langle \Pi_{f_{2}}(\T \etc \T)_i, \s  \rangle \\
%& =  \frac{1}{2}\left (  \langle \Pi_{g_3}(\Phi \etc \Phi)_i, \mathbf{1}_n \rangle +  \langle \Pi_{g_{2}}(\Phi \etc \Phi)_i, 2\s-\mathbf{1}_n \rangle + n \right ).
%\end{align*}
Define 
$$q_{g_i} = \frac{\langle \Pi_{g_3}(\Phi \etc \Phi)_i, \mathbf{1}_n \rangle}{2} + \frac{n}{2} - \frac{\langle \Pi_{g_{2}}(\Phi \etc \Phi)_i,\mathbf{1}_n \rangle}{2}.$$  
Then $\Sigma_f(D)_i = q_{g_i} +  \langle \Pi_{g_{2}}(\Phi \etc \Phi)_i, \s \rangle$. Let $\mathbf{q}_g=(q_{g_1} \etc q_{g_{d^k}})$. We get
\begin{equation} \label{eqn:sigmag} \Sigma_f(D) = \mathbf{q}_g + \Pi_{g_{2}}(\Phi \etc \Phi)\s. \end{equation}

Let $\y$ be the noisy approximation to $\Sigma_f(D)$ released by the privacy mechanism. Given $\y$, the linear program that the adversary solves is:
\begin{align} \label{eqn:newlp}
\mbox{argmin}_{\s}\, \left  \|\y- \mathbf{q}_g - \Pi_{g_{2}}(\Phi \etc \Phi)\s \right  \|_1.
\end{align}
The following theorem analyzes this attack using a random matrix $\Z$. 

\begin{theorem} [Part~\ref{part2}, Theorem~\ref{thm:multlin}] \label{thm:lp}
Let $f: \{0,1\}^{k+1} \to \{0,1\}$ be a non-degenerate boolean function and let $n \leq c d^k/\log_{(q)}d$ for a constant $c$ depending only on $k,q$.  Then there exists a constant $\gamma=\gamma(k,q) > 0$ such that any mechanism which for every database $D \in (\{0,1\})_{n \times (d+1)}$ releases $\Sigma_f(D)$ by adding $o(\sqrt{n})$ (or releases $\Sigma_f(D)/n$ by adding $o(1/\sqrt{n})$) noise to at most $1-\gamma$ fraction of the entries is attribute non-private. The attack that achieves this non-privacy violation runs in $O(d^{2k} n^3 + d^{2.5 k}n^2)$~time.
\end{theorem}
\begin{proof}
%Let us look at privacy noise bounds for releasing $\Sigma_f(D)$.  Let $y_i = \Sigma_f(D)_i +e_i$, where $e_i$ is the noise added by the mechanism while releasing the entry $\Sigma_f(D)_i$. From Equation~\eqref{eqn:sigmag}
%$$\Sigma_f(D) = \mathbf{q}_g + \Pi_{g_{2}}(\Phi \etc \Phi)\s.$$
%The adversary solves $\mbox{argmin}_{\s}  \| \y- \mathbf{q}_g - \Pi_{g_{2}}(\Phi \etc \Phi)\s\|$.
The proof uses the LP decoding attack outlined in Theorem~\ref{thm:d11} on Equation~\eqref{eqn:newlp}. Let $\BT$ be a random matrix of dimension $d \times n$ with independent Bernoulli entries taking values $0$ and $1$ with probability $1/2$, and let database $D=(\BT^\top|\s)$ for some $\s \in \{0,1\}^n$ . Let $\Z= 2\BT - \mathbf{1}_{d \times n}$. To use Theorem~\ref{thm:d11}, we need to \renewcommand{\labelenumi}{(\roman{enumi})}\begin{inparaenum} \item establish that $\Pi_{g_{2}}(\Z \etc \Z)$ is a Euclidean section and \item  establish a lower bound on its least singular value. \end{inparaenum} Since $g_{2} \,:\, \{-1,1\}^{k} \to \{-1,0,1\}$ has a representation as a multilinear polynomial of degree $k$, Theorem~\ref{thm:euclidean} shows that $\Pi_{g_{2}}(\Z \etc \Z)$ is with exponentially high probability a Euclidean section. Repeating an analysis similar to Theorem~\ref{thm:least} shows that the least singular value of $\Pi_{g_{2}}(\Z \etc \Z)$ is with exponentially high probability at least $\Omega(\sqrt{d^{k}})$. Hence, with exponentially high probability both of the following statements hold  simultaneously: (i) $\Pi_{g_{2}}(\Z \etc \Z)$ is a Euclidean section and (ii) $\sigma_n(\Pi_{g_{2}}(\Z \etc \Z)) = \Omega(\sqrt{d^{k}})$. 

Invoking Theorem~\ref{thm:d11} with $\beta=o(\sqrt{n})$, $a=\sqrt{n}$, $b=d^k$, and $\sigma=\Omega(\sqrt{d^k})$ shows that with exponentially high probability the adversary fails to recover only $o(n)$ entries of $\s$. This shows that the mechanism is attribute non-private.

In the running time analysis of Theorem~\ref{thm:d11}, $m$ gets replaced by $d^{k}$, and $N_3$ by $d^k n$ (as the input matrix can be represented using $O(d^k n)$ bits). Noise bound of $o(\sqrt{n})$ for releasing $\Sigma_f(D)$ translates into a noise bound of $o(1/\sqrt{n})$ for releasing $\Sigma_f(D)/n$.  
\end{proof}

\subsection{Spectral and Geometric Properties of Random Row Function Matrices.} \label{sec:spect}
Analysis of our reconstruction attacks rely on spectral and geometric properties of random row function matrices that we discuss in this section. Rudelson~\cite{R11} and Kasiviswanathan~\emph{et al.}\ \cite{KRSU10} analyzed certain spectral and geometric properties of a certain class of correlated matrices that they referred to as row product matrices (or conjunction matrices). Our analysis builds upon these results. We first summarize some important definitions and useful results from~\cite{R11,KRSU10} in Section~\ref{sec:randconj}, and establish our least singular value bound in Section~\ref{sec:leastsing}. The Euclidean section property is established in Appendix~\ref{app:euclid}.

\subsubsection{Spectral and Geometric Properties of Random Row Product Matrices.} \label{sec:randconj}
For two matrices with the same number of columns we define the row product as a matrix whose rows consist of entry-wise product of the rows of original matrices.
\begin{definition} [Row Product Matrix] \label{def:entrywise}
The entry-wise product of vectors $p, q \in\R^n$ is the vector in $p \odot q \in\R^n$ with entries $(p \odot q)_i = p_{i} \cdot q_{i}$. If $\T_{(1)}$ is an $N_1 \times n$ matrix, and $\T_{(2)}$ is an $N_2 \times n$ matrix, denote by $\T_{(1)} \odot \T_{(2)}$ an $N_1N_2 \times n$ matrix, whose rows are entry-wise products of the rows of $\T_{(1)}$ and $\T_{(2)}:$ $(\T_{(1)} \odot \T_{(2)})_{j,k}=\T_{1_j} \odot \T_{2_k}$, where $(\T_{(1)} \odot \T_{(2)})_{j,k},\T_{1_j}, \T_{2_k}$ denote rows of the corresponding matrices.\!\footnote{When $N_1=N_2$, the row product matrix is also called the Hadamard product matrix.}
\end{definition}

Rudelson~\cite{R11} showed that if we take entry-wise product of $k$ independent random matrices of dimension $d \times n$, then the largest and the least singular values of the resulting row product matrix (which of dimension $d^k \times n$) is asymptotically the same order as that of a $d^k \times n$ matrix with i.i.d.\ entries. To formulate this result formally, we introduce a class of uniformly bounded random variables, whose variances are uniformly bounded below.
\begin{definition} [$\tau$-random variable and matrix] \label{defn:tau}
Let $\tau > 0$. We will call a random variable $\xi$ a $\tau$-random variable if $|\xi| \leq 1$ a.s., $\E[\xi] = 0$, and $\E[\xi^2] \geq \tau^2$. A matrix $\mathbf{M}$ is called a $\tau$-random matrix if all its entries are independent $\tau$-random.
\end{definition}

We would also need the notion of iterated logarithm ($\log_{(q)}$) that is defined as: 
\begin{definition} \label{defn:iterlog}
For $q \in  \N$, define the function $\log_{(q)} \, : \, (0,\infty) \to \R$ by induction.
\begin{CompactEnumerate}
\item $\log_{(1)} t = \max(\log t, 1)$;
\item $\log_{(q+1)} t = \log_{(1)}(\log_{(q)} t)$.
\end{CompactEnumerate}
\end{definition}

We are now ready to state the main result from~\cite{R11} that establishes a lower bound on the $\ell_1$ norm of $(\BT_{(1)} \odot \dots \odot \BT_{(k)})\x$ where each $\BT_{(i)}$ is an independent $\tau$-random matrix and $\x$ is a unit vector. 
\begin{theorem}[~\cite{R11}]  \label{thm:R11}
Let $k,q,n,d$ be natural numbers. Assume that $n \le c d^k/\log_{(q)}d.$
Let $\BT_{(1)} \etc \BT_{(k)}$ be $k$ matrices with independent $\tau$-random entries and dimensions $d \times n$. Then the $k$-times entry-wise product $\BT_{(1)}  \odot \BT_{(2)} \odot \dots \odot \BT_{(k)}$ is a $d^k \times n$ matrix satisfying
\begin{equation*}\Pr\left [\exists \x \in S^{n-1} \;\; \|(\BT_{(1)} \odot \dots \odot \BT_{(k)})\x \|_1 \leq C_1 d^k \right ]  \leq c_1 \exp(-C_2 d).\end{equation*}
The constants $c,C_1,c_1,C_2$ depends only on $k$ and $q$.
\end{theorem}	

One of the main ingredients in proving the above theorem is following fact about the norm of row product matrices.
\begin{theorem}[~\cite{R11}] \label{thm:rowprodnorm}
Let $\BT$ be a matrix with independent $\tau$-random entries and dimension $d \times n$. Then the $k$-times entry-wise product $\BT \odot \dots \odot \BT$ is a $d^k \times n$ matrix satisfying
\begin{equation*}\Pr \left [\norm{\BT \odot \dots \odot \BT} \geq c_3 \left (\sqrt{d^k}+\sqrt{n} \right ) \right ]  \leq \exp \left (-c_4 n^{\frac{1}{12k}} \right ).\end{equation*}
The constants $c_3,c_4$ depends only on $k$.
\end{theorem}
The bound on the norm appearing in the above theorem (asymptotically) matches that for a $d^k \times n$ matrix with independent $\tau$-random entries (refer~\cite{V} for more details). 

\subsubsection{Least Singular Value of Random Row Function Matrices.} \label{sec:leastsing}
We start by proving a simple proposition about functions that can be represented as multilinear polynomials. The main step behind the following proposition proof is the following simple fact about multilinear polynomials. Let $P^{(h)}$ be a multilinear polynomial representing function $h$, and let $(\d_1 \etc \d_k) \in \{0,1\}^k$ and $\d_i' \in \{0,1\}$ then
\begin{equation*}
\begin{split}
P^{(h)}(\d_1 \etc \d_i  \etc \d_k) - P^{(h)}(\d_1 \etc \d_i' \etc \d_k)  &  = h(\d_1 \etc \d_i \etc \d_k) - h(\d_1 \etc \d_i' \etc \d_k) \\
& =(\d_i-\d_i') \cdot \frac{\partial\,}{\partial \d_i} P^{(h)}.
\end{split}
\end{equation*}

\begin{proposition}  \label{prop:derivatives}
Let $h$ be a function from  $\{0,1\}^k \to \{-1,0,1\}$ having a representation as a multilinear polynomial of degree $k$.  Let $P^{(h)}$ denote this multilinear polynomial. Let $(\d_1 \etc \d_k), (\d_1' \etc \d_k') \in \{0,1\}^k$. For $I \subseteq \{1 \etc k \}$ let $\d(I) \in \{0,1\}^k$ be the point with coordinates $\d_j(I)=  \d_j' \text{ if } j \in I; \mbox{ and } \d_j(I)= \d_j  \text{ if } j \notin I$. Then 
\begin{equation*}(\d_1-\d_1') \cdot \ldots (\d_k-\d_k') = c_h(k) \sum_{I \subseteq [k]} (-1)^{|I|} h(\d_1(I) \etc \d_k(I)),\end{equation*}
where $1/c_h(k)$ is the coefficient of the monomial corresponding to all $k$ variables in the multilinear representation of $h$.
\end{proposition}
\begin{proof}
By definition, we know that for all $(\d_1 \etc \d_k) \in \{0,1\}^k$, $h(\d_1 \etc \d_k) = P^{(h)}(\d_1 \etc \d_k)$. Since $P^{(h)}$ is a linear function of $\d_1$,
\begin{align*}  P^{(h)}(\d_1, \d_2 \etc \d_k) - P^{(h)}(\d_1', \d_2 \etc \d_k)  & = h(\d_1, \d_2 \etc \d_k) - h(\d_1', \d_2 \etc \d_k) \\ & =(\d_1-\d_1') \cdot \frac{\partial\,}{\partial \d_1} P^{(h)},\end{align*}
where $\frac{\partial\,}{\partial \d_1} P^{(h)}$ denotes the partial derivative of $P^{(h)}(\d_1, \d_2 \etc \d_k)$ with respect to $\delta_1$. Repeating this for the other coordinates, we get
\begin{equation*}
\begin{split}
\sum_{I \subseteq [k]} (-1)^{|I|} h(\d_1(I) \etc \d_k(I)) =  (\d_1-\d_1') \cdot \ldots \cdot (\d_k-\d_k') \cdot 
\left ( \frac{\partial\,}{\partial \d_1} \ldots \frac{\partial\,}{\partial \d_k} P^{(h)} \right ).
\end{split}
\end{equation*}
The last term in the right hand side is the coefficient of the polynomial $P^{(h)}(\d_1 \etc \d_k)$ corresponding to the monomial $\d_1 \cdot \ldots \cdot \d_k$, and we denote it by $1/c_h(k)$. 
\end{proof}

\begin{corollary}[Corollary to Proposition~\ref{prop:derivatives}] \label{cor:derivatives}
Let $\T_{(1)} \etc, \T_{(k)},\T'_{(1)} \etc \T'_{(k)}$ be $2k$ matrices with $\{0,1\}$ entries and dimensions $d \times n$. For a set $I \subseteq [k]$ denote $\T_{(j)}(I)= \T_{(j)} \text{ if } j \in I \mbox{ and } \T_{(j)}(I)=  \T'_{(j)} \text{ if } j \notin I$. Then the following holds,
\begin{equation*}(\T_{(1)}-\T_{(1)}') \odot \ldots \odot (\T_{(k)}-\T_{(k)}')  = c_h(k) \sum_{I \subseteq [k]} (-1)^{|I|} \Pi_h(\T_{(1)}(I) \etc \T_{(k)}(I)).\end{equation*}    
\end{corollary}
\begin{proof}
Follows by a simple extension of Proposition~\ref{prop:derivatives} and using Definitions~\ref{def:funcmat} and~\ref{def:entrywise}. 
\end{proof}

For a random $\BT$, the following theorem shows that the $\ell_1$-norm of $\Pi_h(\BT \etc \BT) \x$ is ``big'' for all $\x \in S^{n-1}$. 
\begin{theorem}  \label{thm:L1}
Let $k,q,n,d$ be natural numbers. Assume that $n \le c d^k/\log_{(q)}d$. Consider a $d \times n$ matrix $\BT$ with independent Bernoulli entries taking values 0 and 1 with probability $1/2$. Let $h$ be a function from $\{0,1\}^k \to \{-1,0,1\} $ having a representation as a multilinear polynomial of degree $k$. Then the matrix $\Pi_h(\BT \etc \BT)$ satisfies 
\begin{equation*}
\Pr\left [ \exists \x \in S^{n-1} \ \norm{\Pi_h(\BT \etc \BT) \x}_1 \le  C' d^k \right ]  \le  c_1 \exp( - c_2 d).
\end{equation*}
The constants $c,C',c_1,c_2$ depend only on $k$ and $q$.
\end{theorem}
\begin{proof}
First notice that if an $a_1 \times n$ matrix $M'$ is formed from the $a_2 \times n$ matrix $M$ by taking a subset of rows, then for any $\x \in \R^n$, $\norm{M'\x}_1 \le \norm{M\x}_1$.

Let $d=2kd'+l$, where $0 \le l<2k$. For $j=1 \etc k$ denote by $\BT_j^1$ the submatrix of $\BT$ consisting of rows $(2d'(j-1)+1) \etc (2d'(j-1)+d')$, and by $\BT_j^0$ the submatrix consisting or rows $(2d'(j-1)+d'+1) \etc 2jd'$. For a set $I \subseteq [k]$ denote
\[ \BT_{(j)}(I)= \begin{cases}
\BT_j^0 &\text{if } j \in I, \\
\BT_j^1 &\text{if } j \notin I.
\end{cases} 
\]
Then Corollary~\ref{cor:derivatives} implies
\begin{align*}
(\BT_1^1-\BT_1^0) \odot \ldots \odot (\BT_k^1-\BT_k^0) = c_h(k) \sum_{I \subseteq [k]} (-1)^{|I|} \Pi_h(\BT_{(1)}(I) \etc \BT_{(k)}(I)).     
\end{align*}
Therefore, by triangle inequality,
\begin{align} \label{eqn:prev}
\norm{(\BT_1^1-\BT_1^0) \odot \ldots \odot (\BT_k^1-\BT_k^0) \x}_1  & \le c_h(k) \sum_{I \subseteq [k]} \norm{\Pi_h(\BT_{(1)}(I) \etc \BT_{(k)}(I)) \x}_1 \\
& \le c_h(k)2^k \norm{\Pi_h(\BT \etc \BT) \x}_1. \nonumber
\end{align}
The last inequality follows since $\forall I$ $\Pi_h(\BT_{(1)}(I) \etc \BT_{(k)}(I))$ is a submatrix of $\Pi_h(\BT \etc \BT)$, which implies that
$$\norm{\Pi_h(\BT_{(1)}(I) \etc \BT_{(k)}(I)) \x}_1  \leq \norm{\Pi_h(\BT \etc \BT) \x}_1.$$
The entries of the matrices $(\BT_1^1-\BT_1^0) \etc (\BT_k^1-\BT_k^0)$ are independent $\tau$-random variables. Thus, Theorem~\ref{thm:R11} (note that $d' = O(d))$ yields
\begin{equation*}
\begin{split}
\Pr[\exists \x \in S^{n-1}\norm{(\BT_1^1-\BT_1^0) \odot \ldots \odot (\BT_k^1-\BT_k^0) \x}_1  \le C d^k  ]  \le  c_1 \exp \left( - c_2 d  \right),
\end{split}
\end{equation*}
which along with Equation~\eqref{eqn:prev} proves the theorem. Note that $c_h(k)$ can be bounded as a function of $k$ alone. 
\end{proof}

Combining Theorem~\ref{thm:L1} with Cauchy-Schwarz's inequality, we obtain a lower bound on the least singular value of $\Pi_h(\BT \etc \BT)$. It is well known~\cite{RV08,V} that the least singular value of a $d^k \times n$ matrix with independent $\tau$-random entries is $\approx d^k$. Therefore, we get that in spite of all the correlations that exist in $\Pi_h(\BT \etc \BT)$ its least singular value is asymptotically the same order as that of an i.i.d.\ matrix.
\begin{theorem} \label{thm:least}
Under the assumptions of Theorem~\ref{thm:L1} 
$$\Pr \left[ \sigma_n(\Pi_h(\BT \etc \BT)) \le C' \sqrt{d^k} \right] \le  c_1 \exp \left( - c_2 d  \right).$$
\end{theorem}
\begin{proof}
By Cauchy-Schwarz inequality, $$\norm{\Pi_h(\BT \etc \BT) \x}_1 \leq \sqrt{d^k} \norm{\Pi_h(\BT \etc \BT) \x}.$$ Therefore,
\begin{align*} \Pr \left[ \exists \x \in S^{n-1} \  \norm{\Pi_h(\BT \etc \BT) \x} \le C' \sqrt{d^k} \right]  \le \Pr \left[ \exists \x \in S^{n-1} \ \norm{\Pi_h(\BT \etc \BT) \x}_1  \le C' d^k \right ].\end{align*}
The right-hand side probability could be bounded using Theorem~\ref{thm:L1} and the left-hand probability is exactly $\Pr \left[ \sigma_n(\Pi_h(\BT \etc \BT)) \le C' \sqrt{d^k} \right]$. 
\end{proof}

\section{Releasing $M$-estimators} \label{sec:mest}
In this section, we analyze privacy lower bounds for releasing $M$-estimators. Assume we have $n$ samples $\x_1 \etc \x_n \in \R^{k+1}$, consider the following optimization problem:
\begin{equation} \label{eqn:mest} \mathcal{L}(\theta;\x_1 \etc \x_n) = \frac{1}{n} \sum_{i=1}^{n} \ell(\theta;\x_i), \end{equation}
where $\theta \in \Theta \subset \R^{k+1}$, the {\em separable} loss function $\mathcal{L} \,: \, \Theta \times (\R^{k+1})^n \to \R$ measures the ``fit'' of $\theta \in \Theta$ to any given data $\x_1 \etc \x_n$, and $\ell \,:\, \Theta \times \R^{k+1} \to \R$ is the loss function associated with a single data point.
It is common to assume that the loss function has certain properties, e.g., for any given sample $\x_1 \etc \x_n$, the loss function assigns a cost $\mathcal{L}(\theta;\x_1 \etc \x_n) \geq 0$ to the estimator $\theta$. The $M$-estimator ($\hat{\theta}$) associated with a given a function $\mathcal{L}(\theta;\x_1 \etc \x_n) \geq 0$ is
\begin{align*}\hat{\theta}  = \mbox{argmin}_{\theta \in \Theta}\, \mathcal{L}(\theta;\x_1 \etc \x_n)  = \mbox{argmin}_{\theta \in \Theta} \, \sum_{i=1}^{n} \ell(\theta;\x_i).\end{align*}
For a differentiable loss function $\ell$, the estimator $\hat{\theta}$ could be found by setting $\partial \, \mathcal{L}(\theta;\x_1 \etc \x_n)$ to zero.

$M$-estimators are natural extensions of the Maximum Likelihood Estimators (MLE)~\cite{stat}. They enjoy similar consistency and are asymptotically normal. natural extensions of the MLE. There are several reasons for studying these estimators: \renewcommand{\labelenumi}{(\roman{enumi})}\begin{inparaenum} \item they may be more computationally efficient than the MLE, and \item they may be more robust (resistant to deviations) than MLE. \end{inparaenum} The linear regression MLE is captured by setting $\ell(\theta;\x) = (y - \langle \z, \theta \rangle)^2$ in Equation~\eqref{eqn:mest} where $\x = (\z,y)$, and MLE for logistic regression is captured by setting $\ell(\theta;\x) = y \langle \z,\theta \rangle - \ell(1+\exp(\z,\theta))$~\cite{micro}.

\mypar{Problem Statement.} Let $D=(\D|\s)$ be a database dimension $n \times d+1$, and let $\D$ be a real-valued matrix of dimension $n \times d$ and $\s \in \{0,1\}^n$ be a (secret) vector. Let $D=(\delta_{i,j})$. Consider the submatrix $D|_J$ of $D$, where $J \in \{1 \etc d\}^{k}$.  Let $\hat{\theta}_J$ be the $M$-estimator for $D|_J$ defined as
\begin{equation*}\hat{\theta}_J = \mbox{argmin}_{\theta \in \Theta} \, \sum_{i=1}^{n} \ell(\theta;(\delta_{i,j_1} \etc \delta_{i,j_k},s_i))  \mbox{ where } J=(j_1 \etc j_k).\end{equation*}
The goal is to understand how much noise is needed to attribute privately release $\hat{\theta}_J$ for all $J$'s when the loss function ($\ell$) is differentiable. 

\mypar{Basic Scheme.} Consider a differentiable loss function $\ell$. Let $D=(\D|\s)$. Let $\D = (\D_{(1)}|  \etc |\D_{(d/k)})$, where each $\D_{(i)}$ is an $n \times k$ matrix (assume $d$ is a multiple of $k$ for simplicity). Consider any $\D_{(i)}$. We have
$$\partial\, \mathcal{L}(\theta;(\D_{(i)},\s)) =  \frac{1}{n} \sum_{j=1}^n \partial \, \ell(\theta_i;(\D_{(i)_j},s_j)),$$
where $\D_{(i)_j}$ is the $j$th row in $\D_{(i)}$. Then $M$-estimators $\hat{\theta}_i$ for $i \in [d/k]$ is obtained by solving
\begin{equation}\label{eqn:gengrad} 
\frac{1}{n} \sum_{j=1}^n \partial \, \ell(\theta_i;(\D_{(i)_j},s_j)) = 0. 
\end{equation}
This gives a set of constraints over $\s$ which the adversary could use to construct $\hat{\s}$. For the case of linear and logistic regression, Equation~\eqref{eqn:gengrad} reduces to a form $\D_{(i)}^\top \s - \mathbf{r}=0$, where $\mathbf{r}$ is a vector independent of $\s$.  For general loss function, we would use the fact that $\s$ is binary and use a decomposition similar to Equation~\eqref{eqn:fdecomp}. The other issue is that the adversary gets only a noisy approximation of $\hat{\theta}_1 \etc \hat{\theta}_{d/k}$  and we overcome this problem by using the Lipschitz properties of the gradient of the loss function.

In the next subsection, we focus on the standard MLE's for linear and logistic regression. In Section~\ref{sec:genfn}, we consider general $M$-estimators. Here, we would require an additional variance condition on the loss function. 
%Missing details from this section are in Appendix~\ref{app:mest}.
We would use this following standard definition of Lipschitz continuous gradient.

\begin{definition} [Lipschitz Continuous Gradient] \label{defn:lip}
The gradient of a function $G \,:\, \R^k \to \R$ is Lipschitz continuous with parameter $\lambda >0$ if, $\|\partial\, G(\x) - \partial\, G(\y) \| \leq \lambda \|\x-\y\|$.
\end{definition}
\begin{remark}
Also, for any twice differentiable function $G$, $\lambda I \succeq  \partial^2\, G(\x)$ for all $\x$ (where $\partial^2\, G(\x)$ denotes the Hessian matrix)~\cite{smola}.
\end{remark}

\subsection{Releasing Linear and Logistic Regression Estimators.} \label{sec:linlog} In this section, we establish distortion lower bounds for attribute privately releasing linear and logistic regression estimators.

\mypar{Linear Regression Background.} A general linear regression problem could be represented as $\s = X\theta + \varepsilon$, where $\s =(s_1 \etc s_n) \in \R^n$ is a vector of observed responses, $X \in \R^{n \times k}$ is a matrix of regressors, $\mathbf{\varepsilon} = (\varepsilon_1 \etc \varepsilon_n)$ is an unobserved error vector where each $\varepsilon_i$ accounts for the discrepancy between the actual responses ($s_i$) and the predicted outcomes ($ \langle X_i, \theta \rangle $), and $\theta \in \R^{k}$ is a vector of unknown estimators. This above optimization problem has a closed-form solution given as $\hat{\theta} = (X^\top X)^{-1} X^\top \s$ (also known as the ordinary least squares estimator).
Let $\LIN(\theta;\s,X,\sigma^2)$ denote the log-likelihood of linear regression (the likelihood expression is given in Appendix~\ref{app:linlogback}). The gradient (w.r.t.\ to $\theta$) of the log-likelihood is given by~\cite{micro}
$$\partial\, \LIN(\theta;\s,X,\sigma^2) = \frac{1}{\sigma^2} \left (X^\top \s - X^\top X\theta \right ).$$

\mypar{Logistic Regression Background.} Logistic regression models estimate probabilities of events as functions of independent variables. Let $s_i$ be binary variable representing the value on the dependent variable for $i$th input, and the values of $k$ independent variables for this same case be represented as $x_{i,j}$ ($j=1\etc k$).  Let $n$ denote the total sample size, and let $\zeta_i$ denote the probability of success $(\Pr[s_i=1]$). The design matrix of independent variables, $X=(x_{i,j})$, is composed of $n$ rows and $k$ columns.

The logistic regression model equates the logit transform, the log-odds of the probability of a success, to the linear component:
\begin{equation*}\ln \left ( \frac{\zeta_i}{1-\zeta_i} \right ) = \sum_{i=1}^k x_{i,j} \theta \equiv \zeta_i = \frac{1+\exp(\langle \theta,X_i \rangle)}{\exp(\langle \theta,X_i \rangle)}  \mbox{ where } X_i \mbox{ is the } i\mbox{th row in } X.\end{equation*}
To emphasize the dependence of $\zeta_i$ on $\theta$, we use the notation $\zeta_i = \zeta_i(\theta)$. Let $\LOG(\theta;\s, X)$ denote the log-likelihood of logistic regression (see Appendix~\ref{app:linlogback} for the likelihood expression). The gradient (w.r.t.\ to $\theta$) of the log-likelihood is given by~\cite{micro}
$$\partial\, \LOG(\theta;\s,X) = X^\top \s -  X^\top \mbox{vert}(\zeta_1(\theta) \etc \zeta_n(\theta)),$$
where the notation $\mbox{vert}(\cdot\etc\cdot)$ denotes vertical concatenation of the argument matrices/vectors. Our analysis will require a bound on the Lipschitz constant of the gradient of the log-likelihood function, which we bound using the following claim.
\begin{claim} \label{claim:lips}
The Lipschitz constant $\Lambda_{\mbox{log}}$ of the gradient of the log-likelihood $\partial\, \LOG(\theta;\s,X)$ can be bounded by the operator norm of $X^\top X$.
\end{claim}
\begin{proof}
From Definition~\ref{defn:lip}, we know that the Lipschitz constant of the gradient of the log-likelihood ($\partial\, \LOG(\theta;\s,X)$), can be bounded by maximum eigenvalue of Hessian of $ \LOG(\theta;\s,X)$. The $(i,j)$th entry of Hessian of $\LOG(\theta;\s,X)$~is
$$ \frac{\partial^2\, \LOG(\theta;\s,X)}{\partial \theta_i \partial \theta_j} = - \sum_{l=1}^n \zeta_\ell(\theta) (1-\zeta_\ell(\theta)) X_{l,i} X_{l,j},$$
where $X_{a,b}$ denote the $(a,b)$th entry in $X$. Note that the Hessian is a $k \times k$ matrix. Since $0 \leq \zeta_\ell(\theta) \leq 1$ (as $\zeta_\ell(\theta)$ represents a probability), we have $$ \frac{\partial^2\, \LOG(\theta;\s,X)}{\partial \theta_i \partial \theta_j} \leq -  \sum_{l=1}^n X_{l,i} X_{l,j}.$$
The Hessian matrix is therefore $-X^\top X$. Hence, the Lipschitz constant ($\Lambda_{\mbox{log}}$) of $\partial\, \LOG(\theta;\s,X)$ can be bounded by the operator norm of $X^\top X$ (see the remark after Definition~\ref{defn:lip}). 
\end{proof}

\mypar{Creating a Linear Reconstruction Problem for Linear Regression.} Let $\D = (\D_{(1)}|  \etc  |\D_{(d/k)})$, where each $\D_{(i)}$ is an $n \times k$ matrix (assume $d$ is a multiple of $k$). Let $\hat{\theta}_i$ be the solution to the MLE equation $$\D_{(i)}^\top \s - \D_{(i)}^\top \D_{(i)}\theta =0.$$ The adversary gets $\tilde{\theta}_i$'s which are a noisy approximation to $\hat{\theta}_i$'s. Given $\tilde{\theta}_1 \etc \tilde{\theta}_{d/k}$, the adversary solves the following set of linear constraints\footnote{Equivalently, we could write the below equation as a single constraint: $$\mbox{vert}(\partial\, \LIN(\tilde{\theta}_1;\s,\D_{(1)},\sigma^2) \etc \partial\, \LIN(\tilde{\theta}_{d/k};\s,\D_{(d/k)},\sigma^2)) =0.$$}  to construct $\hat{\s}$:
\begin{equation*}
\partial\, \LIN(\tilde{\theta}_1;\s,\D_{(1)},\sigma^2) = \cdots =  \partial\, \LIN(\tilde{\theta}_{d/k};\s,\D_{(d/k)},\sigma^2) = 0.
%\mbox{vert}(\partial\, \LIN(\tilde{\theta}_1;\s,\D_{(1)},\sigma^2) \etc \\ \partial\, \LIN(\tilde{\theta}_{d/k};\s,\D_{(d/k)},\sigma^2))  =0. 
\end{equation*}
This can be re-written as:
\begin{equation}
\label{eqn:linear} 
\D_{(1)}^\top\s - \D_{(1)}^\top \D_{(1)} \tilde{\theta}_{1} = \cdots =  \D_{({d/k})}^\top\s - \D_{(d/k)}^\top \D_{(d/k)} \tilde{\theta}_{d/k} = 0.
%\mbox{vert}(\D_{(1)}^\top \etc \D_{({d/k})}^\top) \s  \\
%-\, \mbox{vert}(\D_{(1)}^\top \D_{(1)} \tilde{\theta}_{1} \etc \D_{(d/k)}^\top \D_{(d/k)} \tilde{\theta}_{d/k}) = 0. 
\end{equation}

\mypar{Creating a Linear Reconstruction Problem for Logistic Regression.}  This reduction is very similar to the linear regression case.  As before let $\D = (\D_{(1)}|  \etc  |\D_{(d/k)})$. 
Let $\hat{\theta}_i$ be the solution to the MLE equation $$\D_{(i)}^\top \s - \D_{(i)}^\top\mbox{vert}(\zeta_{(i)_1}(\theta) \etc \zeta_{(i)_n}(\theta)) =0.$$ Let $\zeta_{(i)}(\theta) =\mbox{vert}(\zeta_{(i)_1}(\theta) \etc \zeta_{(i)_n}(\theta))$. The adversary gets $\tilde{\theta}_i$'s which are a noisy approximation to $\hat{\theta}_i$'s. Using $\tilde{\theta}_1 \etc \tilde{\theta}_{d/k}$ and $\D$, the adversary can construct $\zeta_{(1)}(\tilde{\theta}_1) \etc \zeta_{(d/k)}(\tilde{\theta}_{d/k})$. The adversary then solves the following set of linear constraints to construct $\hat{\s}$:
\begin{equation*}
\partial\, \LOG(\tilde{\theta}_1;\s,\D_{(1)})  = \cdots  = \partial\, \LOG(\tilde{\theta}_{d/k};\s,\D_{(d/k)}) =0.
\end{equation*}
This can be re-written as:
\begin{equation} \label{eqn:log} 
\D_{(1)}^\top \s - \D_{(1)}^\top \zeta_{(1)}(\tilde{\theta}_1) =   \cdots  = \D_{({d/k})}^\top \s - \D_{(d/k)}^\top \zeta_{(d/k)}(\tilde{\theta}_{d/k})  = 0. 
\end{equation}

\mypar{Setting up the Least Squares and LP Decoding Attacks.} Consider linear regression. The attacks operate on the linear reconstruction problem of Equation~\eqref{eqn:linear}. The least squares attack constructs $\hat{\s}$ (an approximation of $\s$) by minimizing the  $\ell_2$ norm of the left hand side of Equation~\eqref{eqn:linear}, whereas LP decoding attack works by minimizing $\ell_1$ norm. The attacks are similar for logistic regression except that they now operate on Equation~\eqref{eqn:log}.

In the following analysis, we use a random matrix for $\D$ where each entry of the random matrix is an independent $\tau$-random variable (Definition~\ref{defn:tau}).
\begin{theorem} \label{thm:linlog}
Let $d \geq 2n$ and set $k=1$. Then
\begin{CompactItemize}
\item Any privacy mechanism which for every database $D=(\D|\s)$ where $\D \in \R^{n \times d}$ and $\s \in \{0,1\}^n$ releases the estimators of the linear/logistic regression model between every column of $\D$ and $\s$ by adding $o(1/\sqrt{n})$ noise to each estimator is attribute non-private. The attack that achieves this attribute non-privacy violation runs in  $O(d n^2)$ time.
\item There exists a constant $\gamma > 0$ such that any mechanism which for every database $D=(\D|\s)$ where $\D \in \R^{n \times d}$ and $\s \in \{0,1\}^n$ releases the estimators of the linear/logistic regression model between every column of $\D$ and $\s$ by adding $o(1/\sqrt{n})$ noise to at most $1-\gamma$ fraction of the estimators is attribute non-private. The attack that achieves this attribute non-privacy violation runs in $O(d^2n^3+d^{2.5} n^2)$ time.
\end{CompactItemize}
\end{theorem}
\begin{proof} 
We first do the analysis for linear regression. Let $D=(\M|\s)$, where $\M$ is a $\tau$-random matrix of dimension $n \times d$. \\
\noindent{\textbf{Analysis for Linear Regression.}} Now $\hat{\theta}_i$ is the solution to the MLE equation: $\M_{(i)}^\top \s - \M_{(i)}^\top\M_{(i)} \theta=0$ where $\M_{(i)}$ is the $i$th column of $\M$. Since $k=1$, we have $\hat{\theta}_i \in \R$. The adversary gets noisy approximations of $\hat{\theta}_1 \etc \hat{\theta}_d$. Let $\tilde{\theta}_1 \etc \tilde{\theta}_d$ be these noisy approximations with $\tilde{\theta}_i = \hat{\theta}_i + e_i$ (for some unknown $e_i$). The adversary solves the following set of linear constraints:
\begin{equation*}
\M_{(1)}^\top\s - \M_{(1)}^\top \M_{(1)} \tilde{\theta}_{1} =  \cdots = \M_{(d)}^\top \s  - \M_{(d)}^\top \M_{(d)} \tilde{\theta}_{d} = 0.
\end{equation*}
This could be re-written as:
$$\M^\top\s - \mbox{vert}(\M_{(1)}^\top \M_{(1)} \tilde{\theta}_{1} \etc \M_{(d)}^\top \M_{(d)} \tilde{\theta}_{d}) = 0.$$

Let us first look at the least squares attack. Let $\mathbf{e} = (e_1 \etc e_d)$ (i.e., $\mathbf{e}$ is the error vector). We have $\M_{(i)}^\top \M_{(i)} \leq n$ for all $i \in [d]$ (as $\M$ is a $\tau$-random matrix, all entries in the matrix are at most $1$, and $\M_{(i)}$ is the $i$th column in $\M$). The least squares attack produces an estimate $\hat{\s}$ by solving:
\begin{align*} & \mbox{argmin}_{\s}\, \| \M^\top \s - \mbox{vert}(\M_{(1)}^\top \M_{(1)} \tilde{\theta}_{1} \etc \M_{(d)}^\top \M_{(d)} \tilde{\theta}_{d}) \|. \end{align*}
Since  $\tilde{\theta}_i = \hat{\theta}_i + e_i$, we get for all $i \in [d]$
\begin{align*}\M_{(i)}^\top\M_{(i)} \tilde{\theta}_i &  = \M_{(i)}^\top \M_{(i)} \hat{\theta}_i + \M_{(i)}^\top \M_{(i)} e_i  \leq  \M_{(i)}^\top \M_{(i)} \hat{\theta}_i + n e_i.
\end{align*}
This implies
\begin{equation*} 
\M^\top \s - \mbox{vert}(\M_{(1)}^\top \M_{(1)} \tilde{\theta}_{1} \etc \M_{(d)}^\top \M_{(d)} \tilde{\theta}_{d})
\leq \M^\top \s - \mbox{vert}(\M_{(1)}^\top \M_{(1)} \hat{\theta}_{1} \etc \M_{(d)}^\top \M_{(d)} \hat{\theta}_{d}) + n \cdot \mathbf{e}.
\end{equation*}
The remaining analysis is similar to Theorem~\ref{thm:KRSU}, except that again each error term $e_i$ is scaled up by a factor (at most) $n$. Since  $\M$ is a $\tau$-random matrix, it is well-known that if $d \geq 2n$, then the least singular value of $\M$ is with exponentially high probability $\Omega(\sqrt{d})$ (see, e.g.,~\cite{RV08}). Therefore, if a privacy mechanism adds $o(\sqrt{n})/n$\footnote{The $n$ in the denominator is because of the scaling of the noise $\mathbf{e}$.} noise to each $\hat{\theta}_i$, then the least squares attack with exponentially high probability recovers $1-o(1)$ fraction of the entries in $\s$. The time for executing this attack is $O(dn^2)$ as the attack requires computing the SVD of a $d \times n$ matrix.

The LP decoding attack produces an estimate $\hat{\s}$ by solving:
$$\mbox{argmin}_{\s}\, \| \M^\top\s - \mbox{vert}(\M_{(1)}^\top \M_{(1)} \tilde{\theta}_{1} \etc \M_{(d)}^\top \M_{(d)} \tilde{\theta}_{d}) \|_1.$$
The analysis follows from Theorem~\ref{thm:d11}, except that each error is scaled up by a factor (at most) $n$. Since  $\M$ is a $\tau$-random matrix, it also holds that with exponentially high probability $\M$ is a Euclidean section~\cite{Kas77}. Therefore, if a privacy mechanism adds $o(\sqrt{n})/n=o(1/\sqrt{n})$ noise to at most $1-\gamma$ fraction of the $\hat{\theta}_i$'s, then the LP decoding attack  with exponentially high probability recovers $1-o(1)$ fraction of the entries in $\s$. The time for executing this attack is $O(d^2n^3+d^{2.5} n^2)$ (as there are $d$ constraints, $n$ variables, and the number of the bits in the input is bounded by $dn$). \\
\noindent{\textbf{Analysis for Logistic Regression.}} Consider the estimator $\hat{\theta}_i$ of the logistic regression model between the $i$th column of $\M$ and $\s$. $\hat{\theta}_i$ is the solution to the MLE equation: $\M_{(i)}^\top \s - \M_{(i)}^\top \zeta_{(i)}(\theta) =0$ where $\M_{(i)}$ is the $i$th column of $\M$. Since $k=1$, we have $\hat{\theta}_i \in \R$. The adversary gets noisy approximations of $\hat{\theta}_1 \etc \hat{\theta}_d$. Let $\tilde{\theta}_1 \etc \tilde{\theta}_d$ be these noisy approximations with $\tilde{\theta}_i = \hat{\theta}_i + e_i$. 

Using $\tilde{\theta}_1 \etc \tilde{\theta}_d$ and $\M$, the adversary can construct $\zeta_{(1)}(\tilde{\theta}_1) \etc \zeta_{(d)}(\tilde{\theta}_d)$. The adversary then solves the following set of linear constraints:
\begin{equation*} 
\M_{(1)}^\top\s - \M_{(1)}^\top \zeta_{(1)}(\tilde{\theta}_1)  =  \cdots = \M_{(d)}^\top \s - \M_{(d)}^\top \zeta_{(d)}(\tilde{\theta}_{d})  = 0. \end{equation*}
Let us now apply Lipschitz condition on the gradient of the log-likelihood function. Note that since $k=1$, $\partial \, \LOG(\tilde{\theta}_i;\s,\M_{(i)})$ is a scalar variable (for every $i \in [d]$).
By Lipschitz condition, 
\begin{align*}
|\partial\, \LOG(\tilde{\theta}_i;\s,\M_{(i)}) - \partial\, \LOG(\hat{\theta}_i;\s,\M_{(i)}) |   \leq \Lambda_{\mbox{log}} \,|\tilde{\theta}_i - \hat{\theta}_i | = \Lambda_{\mbox{log}} \, |e_i| ,
\end{align*}
where $\Lambda_{\mbox{log}} = \M_{(i)}^\top \M_{(i)} \leq n$ (using Claim~\ref{claim:lips}). Therefore,
$$ \partial\, \LOG(\tilde{\theta}_i;\s,\M_{(i)}) \leq n |e_i| +  \partial\, \LOG(\hat{\theta}_i;\s,\M_{(i)}). $$
%This means that for $i \in [d]$
%\begin{equation*} 
%\begin{split}
%&\partial\, \LOG(\tilde{\theta}_i;\s,\M_{(i)})  \leq  \partial\, \LOG(\hat{\theta}_i;\s,\M_{(i)}) + n |\mathbf{e}|.
%\end{split}
%\end{equation*}
Substituting $\partial\, \LOG(\tilde{\theta}_i;\s,\M_{(i)}) = \M_{(i)}^\top\s -\M_{(i)}^\top \zeta_{(i)}(\tilde{\theta}_i)$ and $\partial \, \LOG(\hat{\theta}_i;\s,\M_{(i)}) = \M_{(i)}^\top\s -\M_{(i)}^\top \zeta_{(i)}(\hat{\theta}_i)$ in the above equation,
\begin{equation*} 
\begin{split}
&\M_{(i)}^\top\s  - \M_{(i)}^\top \zeta_{(i)}(\tilde{\theta}_i)  \leq \M_{(i)}^\top\s  - \M_{(i)}^\top \zeta_{(i)}(\hat{\theta}_i) + n |e_i|.
\end{split}
\end{equation*}
This implies
\begin{equation*} 
\M^\top \s  - \mbox{vert}(\M_{(1)}^\top \zeta_{(1)}(\tilde{\theta}_1) \etc \M_{(d)}^\top \zeta_{(1)}(\tilde{\theta}_d))
\leq \M^\top \s - \mbox{vert}(\M_{(1)}^\top \zeta_{(1)}(\hat{\theta}_1) \etc \M_{(d)}^\top \zeta_{(d)}(\hat{\theta}_d)) + n \cdot |\mathbf{e}|.
\end{equation*}

The least squares attack produces an estimate $\hat{\s}$ by solving:
\begin{align*}&\mbox{argmin}_{\s} \, \| \M^\top \s  - \mbox{vert}(\M_{(1)}^\top \zeta_{(1)}(\tilde{\theta}_1) \etc \M_{(d)}^\top \zeta_{(d)}(\tilde{\theta}_d)) \|.
\end{align*}
The remaining analysis follows as in the linear regression case (from Theorem~\ref{thm:KRSU}) and again the errors are scaled by a factor (at most) $n$. Therefore, if a privacy mechanism adds $o(\sqrt{n})/n=o(1/\sqrt{n})$ noise to each $\hat{\theta}_i$, then the least squares attack with exponentially high probability recovers $1-o(1)$ fraction of the entries in $\s$. The time for executing this attack is $O(dn^2)$.

The LP decoding attack produces an estimate $\hat{\s}$ by solving:
\begin{align*}&\mbox{argmin}_{\s}\, \| \M^\top \s  - \mbox{vert}(\M_{(1)}^\top \zeta_{(1)}(\tilde{\theta}_1) \etc \M_{(d)}^\top \zeta_{(d)}(\tilde{\theta}_d)) \|_1.
\end{align*}
Again like in the linear regression case (using Theorem~\ref{thm:d11}), we get that any mechanism that adds $o(\sqrt{n})/n=o(1/\sqrt{n})$ noise to at most $1-\gamma$ fraction of the $\hat{\theta}_i$'s is attribute non-private. 
The time for executing this attack is $O(d^2n^3+d^{2.5} n^2)$. 
\end{proof}

Compared to Theorems~\ref{thm:ls} and ~\ref{thm:lp}, this above theorem requires a much larger $d$ (about $O(n)$). However, it is possible to reduce to dependence on $d \approx n^{O(1/k)}$ if the released statistic is a degree $k$ polynomial kernel of these regression functions. We defer the details to the full version.

\subsection{Releasing General $M$-estimators.} \label{sec:genfn}
In this section, we establish distortion lower bounds for attribute privately releasing $M$-estimators associated with differentiable loss functions.\\
\mypar{Creating a Linear Reconstruction Problem.}  Consider Equation~\eqref{eqn:gengrad}. Since $\s$ is a binary vector, we can decompose $\partial\, \ell(\theta;(\D_{(i)_j},s_j)) \in \R^k$ as follows:
\begin{align*}
\partial \, \ell(\theta;(\D_{(i)_j},s_j))  = \ell_0(\theta;\D_{(i)_j})(1-s_j) + \ell_1(\theta;\D_{(i)_j})s_j   = \ell_0(\theta;\D_{(i)_j}) + (\ell_1(\theta;\D_{(i)_j})-\ell_0(\theta;\D_{(i)_j}))s_j,
\end{align*}
where $\ell_0(\theta;\D_{(i)_j}) = \partial\, \ell(\theta;(\D_{(i)_j},0)) \in \R^k$ and $\ell_1(\theta;\D_{(i)_j}) = \partial\, \ell(\theta;(\D_{(i)_j},1)) \in \R^k$.  This is similar to the decomposition in Equation~\eqref{eqn:fdecomp}.  Let 
$$\ell_2(\theta;\D_{(i)_j}) = \ell_1(\theta;\D_{(i)_j})-\ell_0(\theta;\D_{(i)_j}).$$ 
Now the $M$-estimator ($\hat{\theta}_i$) between $\D_{(i)}$ and $\s$ can be found by setting $\partial\, \mathcal{L}(\theta;(\D_{(i)},\s)) =0$. Therefore, $\hat{\theta}_i$ is the solution to  (ignoring the scaling multiplier $1/n$)
\begin{align*} &\partial\, \mathcal{L}(\theta;(\D_{(i)},\s)) = \sum_{j=1}^n \partial \, \ell(\theta_i;(\D_{(i)_j},s_j)) = 0 \\
& \equiv  \sum_{j=1}^n \ell_0(\theta;\D_{(i)_j}) + \ell_2(\theta;\D_{(i)_j})s_j  =0. \end{align*}
The adversary gets $\tilde{\theta}_i$'s which are a noisy approximation to $\hat{\theta}_i$'s. Given $\tilde{\theta}_1 \etc \tilde{\theta}_{d/k}$, the adversary solves the following set of linear constraints:
\begin{align} \label{eqn:glin}
&\partial\, \mathcal{L}(\tilde{\theta}_1;(\D_{(1)},\s)) =  \cdots =  \partial\, \mathcal{L}(\tilde{\theta}_{d/k};(\D_{(d/k)},\s)) = 0 \\
& \equiv \left ( \sum_{j=1}^n \ell_0(\tilde{\theta}_1;\D_{(1)_j}) + \ell_2(\tilde{\theta}_1;\D_{(1)_j})s_j  \right ) =  \cdots  = \left (  \sum_{j=1}^n \ell_0(\tilde{\theta}_{d/k};\D_{(d/k)_j}) + \ell_2(\tilde{\theta}_{d/k};\D_{(d/k)_j})s_j  \right ) = 0. \nonumber
\end{align}
This could also be represented in a matrix-form as we show below.  For every $i \in [d/k]$, define $A_{(i)}$ as a $k \times n$ matrix whose $j$th column is $\ell_0(\tilde{\theta}_i;\D_{(i)_j})$ and $B_{(i)}$ as a $k \times n$ matrix whose $j$th column is $\ell_2(\tilde{\theta}_i;\D_{(i)_j})$. Then  Equation~\eqref{eqn:glin} can be re-written  as
\begin{equation*}
A_{(1)}\mathbf{1}_n + B_{(1)}\s =   \cdots =  A_{(d/k)} \mathbf{1}_n + B_{(d/k)}\s  =0.
\end{equation*}
The adversary solves the above equation to obtain $\hat{\s}$. The analysis uses the following condition on $\ell$.
%\begin{definition}[Variance Condition on Loss\footnote{The assumption that the loss function is centered, i.e., $\E_{\x}[\ell_2(\theta;\x)] =0$ is for convenience. It is possible to deal with non-centered loss functions by suitably adjusting the $\D$ matrices, in which case, the variance condition becomes $\mbox{Var}_\x[\ell_2(\theta;\x)]$ is bounded away from zero.}]
%Consider the decomposition of the gradient of a loss function, $l\,:\, \Theta \times \R^{k+1} \to \R$  as $\partial \ell(\theta;(\x,y)) = \ell_0(\theta;\x) + \ell_2(\theta;\x)y$ where $\theta \in \R^k, \x \in \R^k, y \in \{0,1\}$. The loss function $l$ is said to be satisfy the variance condition if for every $\theta$, $\E_{\x}[\ell_2(\theta;\x)] =0$ and $\E_{\x}[\ell_2(\theta;\x)^2]$ is bounded away from zero.
%\end{definition}

\begin{definition}[Variance Condition]\!\footnote{For the LP decoding attack, we need a stricter condition to achieve the guarantees of Theorem~\ref{thm:d11}. We defer this discussion to the full version.} \label{def:varloss}
Consider the decomposition of the gradient of a loss function, $\ell\,:\, \Theta \times \R^{k+1} \to \R$  as $\partial \ell(\theta;(\x,y)) = \ell_0(\theta;\x) + \ell_2(\theta;\x)y$ where $\theta \in \R^k, \x \in \R^k, y \in \{0,1\}$. The loss function $\ell$ is said to be satisfy the variance condition if for every $\theta$, $\mbox{Var}_{\x}[\ell_2(\theta;\x)]$ is bounded away from zero.
\end{definition}

\begin{theorem} \label{thm:genloss}
Let $\ell$ be a differentiable loss function which satisfies the variance condition. Let $\lambda$ denote the Lipschitz constant of the gradient of the loss function $\ell$. Let $d \geq 2n$ and set $k=1$. Then any privacy mechanism which for every database $D=(\D|\s)$ where $\D \in \R^{n \times d}$ and $\s \in \{0,1\}^n$ releases the $M$-estimators associated with the loss function $\ell$ of the models fitted between every  column of $\D$ and $\s$ by adding $o(1/(\sqrt{n}\lambda))$ noise to each $M$-estimator is attribute non-private. The attack that achieves this attribute non-privacy violation runs in  $O(d n^2)$ time.
%\item There exists a constant $\gamma > 0$ such that any mechanism which for every database $D=(\D|\s)$ where $\D \in \R^{n \times d}$ and $\s \in \{0,1\}^n$ releases the $M$-estimators associated with the loss function $\ell$ of the models fitted between every column of $\D$ and $\s$ by adding $o(\sqrt{n}/\lambda)$ noise to at most $1-\gamma$ fraction of the $M$-estimators is attribute non-private. The attack that achieves this attribute non-privacy violation runs in $O(d^2n^3+d^{2.5} n^2)$ time.
%\end{CompactItemize}
\end{theorem}
\begin{proof}
Let $D=(\M|\s)$, where $\M$ is a $\tau$-random matrix of dimension $n \times d$ .

The $M$-estimator $\hat{\theta}_i \in \R$ between the $i$th column $\M_{(i)}$ of $\M$ and $\s$ is given by the solution to the equation:
$$\sum_{j=1}^n \ell_0(\theta;\M_{(i)_j}) + \ell_2(\theta;\M_{(i)_j})s_j=0.$$
The adversary gets noisy approximations of $\hat{\theta}_1 \etc \hat{\theta}_d$. Let $\tilde{\theta}_1 \etc \tilde{\theta}_d$ be these noisy approximations with $\tilde{\theta}_i = \hat{\theta}_i + e_i$. Consider $\partial \, \mathcal{L}(\tilde{\theta_i};(\M_{(i)},\s))$ and $\partial \, \mathcal{L}(\hat{\theta}_i;(\M_{(i)},\s))$. By Lipschitz condition\!\footnote{The Lipschitz constant $\Lambda$ of $\mathcal{L}$ is at most $n$ times the Lipschitz constant of $\ell$, and therefore $\Lambda \leq n\lambda$.},
\begin{align*}
|\partial \, \mathcal{L}(\tilde{\theta}_i;(\M_{(i)},\s)) - \partial \, \mathcal{L}(\hat{\theta}_i;(\M_{(i)},\s))| \leq \Lambda |\tilde{\theta}_i - \hat{\theta}_i|  \leq \lambda n  |\tilde{\theta}_i - \hat{\theta}_i| = \lambda n |e_i|.
\end{align*}
This implies that
$$\partial \, \mathcal{L}(\tilde{\theta}_i;(\M_{(i)},\s)) \leq \partial \, \mathcal{L}(\hat{\theta}_i;(\M_{(i)},\s)) + \lambda n |e_i| .$$
Substituting for the decomposition of $\ell$ in the above equation gives
\begin{equation} \label{eqn:th}
\sum_{j=1}^n \ell_0(\tilde{\theta}_i;\M_{(i)_j}) + \ell_2(\tilde{\theta}_i;\M_{(i)_j})s_j 
\leq  \sum_{j=1}^n \ell_0(\hat{\theta}_i;\M_{(i)_j}) + \ell_2(\hat{\theta}_i;\M_{(i)_j})s_j + \lambda n |e_i|.
%\lefteqn{\langle (\ell_0(\tilde{\theta}_i;\M_{(i)_1}) \etc \ell_0(\tilde{\theta}_i;\M_{(i)_n})),\mathbf{1}_n \rangle}\\ 
%+ \langle (\ell_2(\tilde{\theta}_i;\M_{(i)_1}) \etc \ell_2(\tilde{\theta}_i;\M_{(i)_n})),\s \rangle \\ 
%\leq \lambda n |e_i| + \langle (\ell_0(\hat{\theta}_i;\M_{(i)_1}) \etc \ell_0(\hat{\theta}_i;\M_{(i)_n})),\mathbf{1}_n \rangle\\
%+ \langle (\ell_2(\hat{\theta}_i;\M_{(i)_1}) \etc \ell_2(\hat{\theta}_i;\M_{(i)_n})),\s \rangle.
\end{equation}

Let $\mathbf{A}_1,\mathbf{B}_2,\mathbf{A}_2,\mathbf{B}_2$ be four matrices of dimension $d \times n$ defined as follows:
\begin{align*}
& \mathbf{A}_1 :  \mbox{ $i$th row of $\mathbf{A}_1$ is } \ell_0(\hat{\theta}_i;\M_{(i)_1}) \etc \ell_0(\hat{\theta}_i;\M_{(i)_n}), \\
& \mathbf{B}_1 :  \mbox{ $i$th row of $\mathbf{B}_1$ is  } \ell_2(\hat{\theta}_i;\M_{(i)_1}) \etc \ell_2(\hat{\theta}_i;\M_{(i)_n}), \\
& \mathbf{A}_2 : \mbox{ $i$th row of $\mathbf{A}_2$ is } \ell_0(\tilde{\theta}_i;\M_{(i)_1}) \etc \ell_0(\tilde{\theta}_i;\M_{(i)_n}), \\
& \mathbf{B}_2 : \mbox{ $i$th row of $\mathbf{B}_2$ is } \ell_2(\tilde{\theta}_i;\M_{(i)_1}) \etc \ell_2(\tilde{\theta}_i;\M_{(i)_n}).
\end{align*}

The adversary solves the following reconstruction problem to compute $\hat{\s}$:
\begin{equation} \label{eqn:A2B2} \mathbf{A}_2 \mathbf{1}_n + \mathbf{B}_2 \s = 0. \end{equation}
From Equation~\eqref{eqn:th} it follows that
$$\mathbf{A}_2 \mathbf{1}_n + \mathbf{B}_2 \s \leq \mathbf{A}_1 \mathbf{1}_n + \mathbf{B}_1 \s + \lambda n  |\mathbf{e}|.$$

\noindent{\textbf{$\mathbf{\mathbf{B}_2}$: Least Singular Value.}} Since $\M$ is a $\tau$-random matrix, $\mathbf{B}_2$ is another random matrix. However, $\mathbf{B}_2$ may not be centered (i.e., its entries might have non-zero means). We can re-express $\mathbf{B}_2$ as
$$\mathbf{B}_2 = \underbrace{\mathbf{B}_2 - \E[\mathbf{B}_2]}_{\mathbf{R}} + \E[\mathbf{B}_2].$$
Here, $\mathbf{R}$ is a $\tau'$-random matrix\!\footnote{The expected value of each entry in $\mathbf{R}$ is zero and due to variance condition on the loss function (Definition~\ref{def:varloss}), the second moment of each entry of $\mathbf{R}$ is bounded away from zero. We can re-scale Equation~\eqref{eqn:A2B2} to ensures that all the entries in $\mathbf{B}_2$ have absolute value less than $1$ without changing the solution $\hat{\s}$. This implies that $\mathbf{R}$ is $\tau'$-random matrix for an appropriate $\tau'$ depending on $\ell_2,\tilde{\theta}_1\etc\tilde{\theta}_d$.}  and $\E[\mathbf{B}_2]$ is a matrix of the form $\mathbf{u} \mathbf{v}^\top$ where $\mathbf{u}$ is a $d$-dimensional column vector with entries $\E_{x}[\ell_2(\tilde{\theta}_1;x)] \etc \E_{x}[\ell_2(\tilde{\theta}_d;x)]$ and $\mathbf{v}$ is a $n$-dimensional column vector of all ones. 
\begin{claim}
$\Pr[\sigma_n(\mathbf{B}_2)  \leq c_{19}\sqrt{d} ] \leq \exp(-c_{20}d)$.
\end{claim}
\begin{proof}
$\mathbf{B}_2 = \mathbf{R} + \mathbf{u}\mathbf{v}^\top$, where $\mathbf{R}$ is a $\tau'$-random matrix. The rank of $\mathbf{u}\mathbf{v}^\top$ is $1$, and its operator norm can be polynomially bounded in $d$. Applying Lemma~\ref{lem:randpert} implies the result. 
\end{proof}
\noindent{\textbf{Least Squares Attack.}} The least squares attack produces an estimate $\hat{\s}$ by solving:
$$\hat{\s} = \mbox{argmin}_\s  \, \| \mathbf{A}_2 \mathbf{1}_n + \mathbf{B}_2 \s \|.$$
The analysis is similar to Theorem~\ref{thm:KRSU}, except that each error term $e_i$ is scaled up by a factor (at most) $\lambda n$. Therefore, if a privacy mechanism adds $o(1/(\sqrt{n}\lambda))$ noise to each $\hat{\theta}_i$, then the least squares attack with exponentially high probability recovers $1-o(1)$ fraction of the entries in $\s$. The time for executing this attack is $O(dn^2)$. 

%The LP decoding attack produces an estimate $\hat{\s}$ by solving:
%$$\hat{\s} = \mbox{argmin}_\s  \; \| \mathbf{A}_2 \mathbf{1}_n + \mathbf{B}_2 \s \|_1.$$
%Therefore, if a privacy mechanism adds $o(\sqrt{n}/\lambda)$ noise to at most $1-\gamma$ fraction of the $\hat{\theta}_i$'s, then the LP decoding attack with exponentially high probability recovers $1-o(1)$ fraction of the entries in $\s$. The time for executing this attack is $O(d^2n^3+d^{2.5} n^2).$ 
\end{proof}
 
%\begin{small}
%\bibliographystyle{acm}
%\bibliography{pacparity}

\begin{thebibliography}{10}

\bibitem{De11}
{\sc De, A.}
\newblock {Lower Bounds in Differential Privacy}.
\newblock In {\em TCC\/} (2012), pp.~321--338.

\bibitem{DiDwNi}
{\sc Dinur, I., Dwork, C., and Nissim, K.}
\newblock Revealing information while preserving privacy, full version
  of~\cite{DiNi03}, in preparation, 2009.

\bibitem{DiNi03}
{\sc Dinur, I., and Nissim, K.}
\newblock Revealing information while preserving privacy.
\newblock In {\em PODS\/} (2003), ACM, pp.~202--210.

\bibitem{Dwork06}
{\sc Dwork, C.}
\newblock Differential privacy.
\newblock In {\em ICALP\/} (2006), LNCS, pp.~1--12.

\bibitem{DMNS06}
{\sc Dwork, C., McSherry, F., Nissim, K., and Smith, A.}
\newblock Calibrating noise to sensitivity in private data analysis.
\newblock In {\em TCC\/} (2006), vol.~3876 of {\em LNCS}, Springer,
  pp.~265--284.

\bibitem{DMT07}
{\sc Dwork, C., McSherry, F., and Talwar, K.}
\newblock The price of privacy and the limits of lp decoding.
\newblock In {\em STOC\/} (2007), ACM, pp.~85--94.

\bibitem{DY08}
{\sc Dwork, C., and Yekhanin, S.}
\newblock New efficient attacks on statistical disclosure control mechanisms.
\newblock In {\em CRYPTO\/} (2008), Springer, pp.~469--480.

\bibitem{GRS09}
{\sc Ghosh, A., Roughgarden, T., and Sundararajan, M.}
\newblock Universally utility-maximizing privacy mechanisms.
\newblock In {\em STOC\/} (2009).

\bibitem{HT10}
{\sc Hardt, M., and Talwar, K.}
\newblock {On the Geometry of Differential Privacy}.
\newblock In {\em STOC\/} (2010), ACM, pp.~705--714.

\bibitem{Kas77}
{\sc Kashin, B.~S.}
\newblock {Diameters of some finite-dimensional sets and classes of smooth
  functions}.
\newblock {\em Izv. Akad. Nauk SSSR Ser. Mat. 41\/} (1977), 334--351.

\bibitem{KRSU10}
{\sc Kasiviswanathan, S.~P., Rudelson, M., Smith, A., and Ullman, J.}
\newblock {The Price of Privately Releasing Contingency Tables and the Spectra
  of Random Matrices with Correlated Rows}.
\newblock In {\em STOC\/} (2010), pp.~775--784.

\bibitem{micro}
{\sc Lee, M.-j.}
\newblock {\em {Micro-Econometrics: Methods of Moments and Limited Dependent
  Variables}}.
\newblock Springer, 2010.

\bibitem{MMPRTV10}
{\sc McGregor, A., Mironov, I., Pitassi, T., Reingold, O., Talwar, K., and
  Vadhan, S.~P.}
\newblock The limits of two-party differential privacy.
\newblock In {\em FOCS\/} (2010), IEEE Computer Society, pp.~81--90.

\bibitem{MuthuN12}
{\sc Muthukrishnan, S., and Nikolov, A.}
\newblock Optimal private halfspace counting via discrepancy.
\newblock In {\em STOC\/} (2012), H.~J. Karloff and T.~Pitassi, Eds., ACM,
  pp.~1285--1292.

\bibitem{boolfunc}
{\sc O'Donnell, R.}
\newblock Analysis of boolean functions, 2012.

\bibitem{stat}
{\sc Pfanzagl, J.}
\newblock {\em Parametric Statistical Theory}.
\newblock WDeG Press, 1991.

\bibitem{R11}
{\sc {Rudelson}, M.}
\newblock {Row Products of Random Matrices}.
\newblock {\em ArXiv e-prints\/} (2011).

\bibitem{RV}
{\sc Rudelson, M., and Vershynin, R.}
\newblock {The Littlewood--Offord problem and invertibility of random
  matrices}.
\newblock {\em Advances in Mathematics 218}, 2 (2008), 600--633.

\bibitem{RV08}
{\sc Rudelson, M., and Vershynin, R.}
\newblock {Smallest Singular Value of a Random Rectangular Matrix}.
\newblock {\em Communications on Pure and Applied Mathematics 62}, 12 (2009),
  1707--1739.

\bibitem{smola}
{\sc Smola, A., and Vishwanathan, S.}
\newblock {\em {Introduction to Machine Learning}}.
\newblock \url{http://alex.smola.org/drafts/thebook.pdf}, 2012.

\bibitem{Vaidya}
{\sc Vaidya, P.~M.}
\newblock An algorithm for linear programming which requires $o(((m+n)n^2 +
  (m+n)^{1.5}n)l)$ arithmetic operations.
\newblock {\em Math. Program. 47}, 2 (1990), 175--201.

\bibitem{V}
{\sc Vershynin, R.}
\newblock {Introduction to the Non-Asymptotic Analysis of Random Matrices},
  2011.

\end{thebibliography}
%\end{small}

\appendix
%\section{Analysis of the Least Squares Attack} \label{app:leastsq}
\section{Preliminaries about Boolean Functions} \label{app:boolprelim}
We start with an alternate definition of non-degeneracy and show that it is equivalent to Definition~\ref{def:deg1}.
\begin{definition} \label{def:deg}
A boolean function $f: \{0,1\}^{k+1} \to \{0,1\}$ is called non-degenerate if
\begin{align*} \sum_{(\d_1 \etc \d_{k+1}) \in \{0,1\}^{k+1}} (-1)^{f(\d_1 \etc \d_{k+1})-\sum_{j=1}^{k+1} \d_j} \neq 0. \end{align*}
\end{definition}

Any function on the discrete cube $\{-1,1\}^{k+1}$ can be decomposed into a linear combination of characters, which are Walsh functions. Such representation allows to extend the function $f$ from the discrete cube to $\R^{k+1}$ as a multilinear (i.e., linear with respect to each variable separately) polynomial. In what follows, we will use this extension. The following lemma shows that the non-degeneracy condition is equivalent to the fact that this polynomial has the maximal degree.
\begin{lemma} \label{lem:deg}
Definitions~\ref{def:deg} and~\ref{def:deg1} are equivalent.
%A function $f: \{0,1\}^{k+1} \to \{0,1\}$ is non-degenerate iff it can be written as a multilinear polynomial of degree $k+1$.
\end{lemma}
\begin{proof}
Consider the function $g: \{-1,1\}^{k+1} \to \{-1,1\}$ defined by
\begin{equation*} g(\phi_1 \etc \phi_{k+1}) =2f \left( \frac{1}{2} (1+ \phi_1) \etc \frac{1}{2} (1+ \phi_{k+1})  \right) -1. 
\end{equation*}
Let $\phi=(\phi_1 \etc \phi_{k+1})$. For $S \subseteq\{1 \etc {k+1}\}$ let $\chi_S(\phi)=\prod_{j \in S} \phi_j$ be the corresponding Walsh function. Then the function $g$ can be decomposed as $$g(\phi)= \sum_{S \subseteq\{1 \etc k+1\}} \widehat{g}(S) \chi_S(\phi).$$ Note that $\text{deg}(f)=\text{deg}(g)$, and so $\text{deg}(f)=k+1$ iff $\widehat{g}(1 \etc k+1) \neq 0$. We have
\begin{equation}  \label{coefficient}
\widehat{g}(1 \etc k+1)=2^{-k-1} \sum_{\phi \in \{-1,1\}^{k+1}} g(\phi) \prod_{j=1}^{k+1} \phi_j.
\end{equation}
Since $$g(\phi)=(-1)^{f(\d_1 \etc \d_{k+1})}\;\; \mbox{ and } \;\; \prod_{j=1}^{k+1} \phi_j=(-1)^{\sum_{j=1}^{k+1} \d_j},$$ 
where $\d_j=(1-\phi_j)/2$, the lemma follows. 
\end{proof}

\begin{remark}
There are 10 non-degenerate functions of two boolean variables:
\begin{align*}
& \mbox{AND: }\d_1 \wedge \d_2, \ \d_1 \wedge (\neg \d_2), \  (\neg \d_1) \wedge \d_2, \ ( \neg \d_1) \wedge (\neg \d_2), \\
& \mbox{OR: }\d_1 \vee \d_2, \ \d_1 \vee (\neg \d_2), \  (\neg \d_1) \vee \d_2, \ ( \neg \d_1) \vee (\neg \d_2),\\
& \mbox{XOR: } \d_1 \oplus \d_2, \ (\neg \d_1) \oplus (\neg \d_2).
\end{align*}
The remaining $6$ boolean functions of two variables are degenerate.
\end{remark}

\section{Euclidean Section Property of Random Row Function Matrices} \label{app:euclid}
In this section, we establish the Euclidean section property needed for Theorem~\ref{thm:lp} (LP decoding attack). Let us consider a function $h \,:\, \{-1,1\}^k \to \{-1,0,1\}$ having a representation as a multilinear polynomial of degree $k$. Let $P^{(h)}$ denote this multilinear polynomial. We first prove a proposition analogous to Proposition~\ref{prop:derivatives} for $h$. For $I \subseteq [k]$, let us define 
\begin{equation*}P^{(h)}_{I}(\phi_1 \etc \phi_k) = P^{(h)}(\phi_1' \etc \phi_k')   \mbox{ where } \phi_i' = \phi_i \mbox{ if } i \in I, \mbox{ else } \phi_i' = 0.\end{equation*}
That is $P^{(h)}_{I}(\phi_1 \etc \phi_k)$ is the multilinear polynomial $P_h(\phi_1 \etc \phi_k)$ restricted to variables only in $I$. Under this notation
\[ P^{(h)}_{[k]}(\phi_1,\phi_2 \etc \phi_k) = P^{(h)}(\phi_1,\phi_2 \etc \phi_k).\]

\begin{proposition} \label{prop:plus1}
Let $h$ be a function from  $\{-1,1\}^k \to \{-1,0,1\}$ having a representation as a multilinear polynomial of degree $k$. Let $P^{(h)}$ be this multilinear polynomial. Let $(\phi_1 \etc \phi_k) \in \{-1,1\}^k$. Then 
\begin{equation*}  \phi_1 \cdot \ldots \cdot \phi_k  = c_h(k) ( \sum_{I \subset [k]} (-1)^{k-|I|} P^{(h)}_{I}(\phi_1 \etc \phi_k) + h(\phi_1 \etc \phi_k) ), \end{equation*}
where $1/c_h(k)$ is the coefficient of the monomial corresponding to all $k$ variables in $P^{(h)}$.
\end{proposition}
\begin{proof}
The proof is similar to Proposition~\ref{prop:derivatives}. Since $P^{(h)}$ is a linear function in $\phi_1$,
\begin{align*} P^{(h)}(\phi_1,\phi_2 \etc \phi_k) - P^{(h)}(0,\phi_2 \etc \phi_k) & = P^{(h)}_{[k]}(\phi_1,\phi_2 \etc \phi_k) - P^{(h)}_{\{2 \etc k\}}(\phi_1,\phi_2 \etc \phi_k) \\
& = \phi_1 \frac{\partial\,}{\partial \phi_1} P^{(h)},\end{align*}
where $\frac{\partial\,}{\partial \phi_1} P^{(h)}$ is the partial derivative of $P^{(h)}(\phi_1 \etc \phi_k)$ with respect to $\phi_1$. Repeating this for other coordinates, we get
\begin{equation*} \sum_{I \subseteq [k]} (-1)^{k-|I|} P^{(h)}_{I}(\phi_1 \etc \phi_k) =  \phi_1 \cdot \ldots \cdot \phi_k \cdot  \left ( \frac{\partial\,}{\partial \phi_1} \ldots \frac{\partial\,}{\partial \phi_k} P^{(h)} \right ).\end{equation*}
When $I=[k]$, $$h(\phi_1 \etc \phi_k) = P^{(h)}_{I}(\phi_1 \etc \phi_k),$$ so the above equation could be re-expressed as
\begin{equation*} \sum_{I \subset [k]} (-1)^{k-|I|} P^{(h)}_{I}(\phi_1 \etc \phi_k) + h(\phi_1 \etc \phi_k) = 
\phi_1 \cdot \ldots \cdot \phi_k \cdot 
\left ( \frac{\partial\,}{\partial \phi_1} \ldots \frac{\partial\,}{\partial \phi_k} P^{(h)} \right ). \\
\end{equation*}
The last term in the right hand side is the coefficient of the polynomial $P^{(h)}(\phi_1 \etc \phi_k)$ corresponding to the monomial $\phi_1 \cdot \ldots \cdot \phi_k$, and we denote it by $1/c_h(k)$. 
\end{proof}
Let $\Phi_{(1)} = (\phi_{i,j}^{(1)}) \etc \Phi_{(k)}= (\phi_{i,j}^{(k)})$ be $k$ matrices with $\{-1,1\}$ entries and dimensions $d \times n$. Let us define a $d^k \times n$ matrix $\Pi_{P^{(h)}_{I}}(\Phi_{(1)} \etc \Phi_{(k)})$ as in Definition~\ref{def:funcmat} using the multilinear polynomial $P^{(h)}_{I}$. More concretely, for $J = (j_1, j_2 \etc j_k) \in \{1 \etc d\}^k$ the $(J,a)$ entry of the matrix $\Pi_{P^{(h)}_{I}}(\Phi_{(1)} \etc \Phi_{(k)})$ will be defined by the relation
\[ \pi_{J,a}=P^{(h)}_{I}(\phi^{(1)}_{j_1,a}, \phi^{(2)}_{j_2,a} \etc \phi^{(k)}_{j_k,a}). \] 
Using this definition, the following corollary follows immediately from Proposition~\ref{prop:plus1}.
\begin{corollary} \label{cor:pmderivatives}
Let $\Phi_{(1)} \etc, \Phi_{(k)}$ be $k$ matrices with $\{-1,1\}$ entries and dimensions $d \times n$. Then 
\begin{equation*}\Phi_{(1)} \odot \ldots \odot \Phi_{(k)} = c_h(k)  \left (\sum_{I \subset [k]} (-1)^{k-|I|} \Pi_{P^{(h)}_I}(\Phi_{(1)} \etc \Phi_{(k)})   +  \Pi_{h}(\Phi_{(1)} \etc \Phi_{(k)}) \right ).\end{equation*}
\end{corollary}

\begin{theorem} \label{thm:norm}
Let $k,n,d$ be natural numbers. Consider a $d \times n$ matrix $\Z$ with independent Bernoulli entries taking values $-1$ and $1$ with probability $1/2$. Let $h$ be a function from  $\{-1,1\}^k \to \{-1,0,1\}$ having a representation as a multilinear polynomial of degree $k$. Then the matrix $\Pi_h(\Z \etc \Z)$ satisfies
\begin{equation*} \Pr \left[ \norm{\Pi_h(\Z \etc \Z)}  \ge  c_6 \left (\sqrt{d^k}+\sqrt{n} \right ) \right]  \le c_7 \exp \left( -c_8 \left (n^\frac{1}{12k} \right )  \right). \end{equation*}
The constants $c_6,c_7,c_8$ depend only on $k$.
\end{theorem}
\begin{proof}
From Corollary~\ref{cor:pmderivatives}, we know that
\begin{equation}\label{eqn:norm} 
\Pi_{h}(\Z \etc \Z) =  \Pi_{P^{(h)}_{[k]}}(\Z \etc \Z) 
=\frac{\Z \odot \ldots \odot \Z}{c_h(k)} -\sum_{I \subset [k]} (-1)^{k-|I|} \Pi_{P^{(h)}_I}(\Z \etc \Z). \end{equation} 
To prove the theorem, we use induction over the size of $I$. Our inductive claim is that for every $I \subseteq [k]$,
\begin{equation*}\Pr \left[ \norm{\Pi_{P^{(h)}_I}(\Z \etc \Z)} \geq c_6 \left (d^{|I|/2}+\sqrt{n} \right ) \right] \le c_7 \exp \left( -c_8  n^\frac{1}{12|I|}   \right),\end{equation*}
where constants $c_6,c_7,c_8$ depend only on $|I|$.

\noindent{\textbf{Step 1.}} Let $|I| = 1$. Then $$\Pi_{P^{(h)}_I}(\Z \etc \Z) = c_9 \Z$$ for some constant $c_9$. Therefore, $$\norm{\Pi_{P^{(h)}_I}(\Z \etc \Z)} = c_9\norm{\Z}.$$ 
Since  $\Z$ is a random $\{-1,1\}$ matrix, it is well known that (see e.g.,~\cite{V}) there exists constant $c_{10},c_{11}$ such that
\begin{equation*}\Pr[\norm{\Z} \geq c_{10}(\sqrt{d}+\sqrt{n}) ]  \leq \exp(-c_{11} n).\end{equation*}
Therefore, 
\begin{equation*}\Pr \left [\norm{\Pi_{P^{(h)}_I}(\Z \etc \Z)} \geq c_9 c_{10} (\sqrt{d}+\sqrt{n}) \right ] \\ \leq \exp \left (-c_{11} n \right).\end{equation*}
Therefore, there exists constants $c_6,c_7,c_8$ such that
\begin{equation*} \Pr \left [\norm{\Pi_{P^{(h)}_I}(\Z \etc \Z)} \geq c_6 \left (\sqrt{d}+\sqrt{n} \right ) \right ] \\ \leq c_7\exp \left (-c_8 n \right ) \leq c_7\exp \left (-c_8 n^\frac{1}{12} \right ).\end{equation*}
This completes the basis for induction.

\noindent{\textbf{Step 2.}} Let $|I| = k$, i.e., $I=\{1 \etc k \} = [k]$. We want to bound $\norm{\Pi_{P^{(h)}_{I}}(\Z \etc \Z)}$. By inductive hypothesis, we assume that for every $L \subset [k]$,
\begin{equation} \label{eqn:one}
\Pr \left [\norm{\Pi_{P^{(h)}_{L}}(\Z \etc \Z)} \geq c_6 \left (d^{|L|/2}+\sqrt{n} \right ) \right ]  \leq  c_7 \exp \left( -c_8 \left (n^\frac{1}{12|L|} \right ) \right ), 
\end{equation}
where constants $c_6,c_7,c_8$ depend only on $|L|$. From Theorem~\ref{thm:rowprodnorm}, we know that the $k$-times entry wise product $\Z \odot \dots \odot \Z$ satisfies  the following norm condition (as $\Z$ is matrix with independent $\tau$-random entries)
\begin{equation} \label{eqn:two} \Pr \left [\norm{\Z \odot \dots \odot \Z} \geq c_3 \left (\sqrt{d^k}+\sqrt{n} \right ) \right ]  \leq \exp \left (-c_4 n^\frac{1}{12 k} \right ),\end{equation}
where again the constants $c_3,c_4$ depend only on $k$. %Therefore, there exists a constant $c_{12}=c_3/c_h(k)$ such that
%\begin{align} \label{eqn:two} \Pr \left  [\norm{\frac{\Z \odot \dots \odot \Z}{c_h(k)}} \geq c_{12} \left (\sqrt{d^k}+\sqrt{n} \right ) \right ] \leq \exp \left (-c_4\left (d+ \frac{n}{d^{k-1}} \right ) \right ).\end{align}
From Equations~\eqref{eqn:one} and~\eqref{eqn:two}, we get that there exists constants $c_6,c_7,c_8$ (depending only on $k$) such that
\begin{equation} \label{eqn:last}
\Pr [ \sum_{L \subset [k]} \norm{\Pi_{P^{(h)}_{L}}(\Z \etc \Z)} + \norm{\Z \odot \dots \odot \Z}   \geq c_6 \left (\sqrt{d^k}+\sqrt{n} \right ) ]  
\leq c_7\exp \left (-c_8 n^\frac{1}{12 k}  \right ).
\end{equation}
From Equation~\eqref{eqn:norm},
\begin{align*}
& \Pr \left [ \norm{\Pi_{P^{(h)}_I}(\Z \etc \Z)}  \ge  c_6 \left (\sqrt{d^k}+\sqrt{n} \right ) \right ] \\
&  = \Pr  \left [\| \Z \odot \ldots \odot \Z - \sum_{L \subset [k]} (-1)^{k-|L|} \Pi_{P^{(h)}_L}(\Z \etc \Z) \|  \geq   c_6 (\sqrt{d^k}+\sqrt{n}  ) \right ] \\
& \leq c_7 \exp \left (-c_8 \left ( n^{\frac{1}{12k}} \right ) \right ).
\end{align*}
The last inequality comes by applying triangle inequality and then using  from Equation~\eqref{eqn:last}. Note that for $I =[k]$, $\norm{\Pi_{P^{(h)}_I}(\Z \etc \Z)} = \norm{\Pi_{h}(\Z \etc \Z)}$. This completes the proof of the theorem. 
\end{proof}

\begin{theorem}\label{thm:euclidean}
Let $k,q,n,d$ be natural numbers. Assume that $n \le c d^k/\log_{(q)}d$. Consider a $d \times n$ matrix $\Z$ with independent Bernoulli entries taking values $-1$ and $1$ with probability $1/2$. Let $h$ be a function $\{-1,1\}^k \to \{-1,0,1\}$  having a representation as a multilinear polynomial of degree $k$. Then the matrix $\Pi_h(\Z \etc \Z)$ is with exponentially high probability a Euclidean section.
\end{theorem}
\begin{proof}
Firstly note that $\Pi_h(\Z \etc \Z)$ is an operator from $\R^{n} \to \R^{d^k}$. As mentioned before, by Cauchy-Schwarz's inequality
\begin{equation*} \forall \x \in S^{n-1} \;\; \sqrt{d^k} \norm{\Pi_h(\Z \etc \Z)\x}_2  \\ \geq \norm{\Pi_h(\Z \etc \Z) \x}_1.\end{equation*}

Note that the proof idea of Theorem~\ref{thm:L1} also works for $h$. That is  if $n \le c d^k/\log_{(q)} d$ then
\begin{equation*} \Pr \left[ \exists \x \in S^{n-1} \; \norm{\Pi_h(\Z \etc \Z) \x}_1  \le  C' d^{k}   \right]  \le  c_1 \exp \left( - c_2 d  \right). \end{equation*}
In other words,
\begin{equation*} \Pr \left[ \forall \x \in S^{n-1} \; \norm{\Pi_h(\Z \etc \Z) \x}_1 \ge  C' d^{k}   \right]  \\ \ge  1-c_1 \exp \left( - c_2 d  \right). \end{equation*}
Theorem~\ref{thm:norm} implies that 
\begin{equation*} \begin{split} & \Pr [ \forall \x \in S^{n-1} \; \norm{\Pi_h(\Z \etc \Z)\x}_2  \le  c_6 \left (\sqrt{d^k} +\sqrt{n} \right ) ]  \\ & \ge 1- c_7\exp \left( -c_8 \left (n^\frac{1}{12k} \right )  \right). \end{split} \end{equation*}
If $n \le c d^k/\log_{(q)} d$, then there exists a constant $c_{13}$ (depending only on $k$) such that 
\begin{equation*} \Pr \left[ \forall \x \in S^{n-1} \; \norm{\Pi_h(\Z \etc \Z)\x}_2  \le  c_{13} \sqrt{d^k}  \right]   \ge 1- c_7\exp \left( -c_8 \left (n^\frac{1}{12k} \right )  \right). \end{equation*}
Therefore, with probability at least  $1-c_1 \exp \left( - c_2 d  \right)-c_7\exp \left( -c_8 \left (n^\frac{1}{12k} \right )  \right)$, there exists a $\alpha$ (depending only on $k$ and $q$) such that
\begin{equation*} \forall \x \in S^{n-1} \;\; \norm{\Pi_h(\Z \etc \Z) \x}_1   \geq \alpha \sqrt{d^k} \norm{\Pi_h(\Z \etc \Z)\x}_2.\end{equation*}
This shows that the matrix $\Pi_h(\Z \etc \Z)$ is with exponentially high probability a Euclidean section. 
\end{proof}

\section{MLE's for Linear and Logistic Regression} \label{app:linlogback}
Consider the linear regression problem $\s = X\theta + \varepsilon$. The likelihood function for linear regression (under the assumption that the entries in $\s$ are normally distributed)~is: 
$$\prod_{i=1}^{n} \frac{1}{\sqrt{2 \pi \sigma^2}} \exp \left (-\frac{s_i - \langle X_i, \theta \rangle}{2\sigma^2} \right ),$$
where $X_i$ is the $i$th row in $X$.
The log-likelihood is:
\begin{align*}
\LIN(\theta;\s,X,\sigma^2) & = \sum_{i=1}^{n} \ln (\frac{1}{\sqrt{2\pi \sigma^2}} \exp (-\frac{s_i - \langle X_i, \theta \rangle}{2\sigma^2} ) ) \\
& = -\frac{n}{2}\ln(2\pi) - \frac{n}{2} \ln \sigma^2 -\frac{1}{2} \left (\frac{(\s-X\theta)^\top(\s-X\theta)}{\sigma^2} \right ).
\end{align*}

For the logistic regression problem, the likelihood function (under the assumption that the entries in $\s$ are binomially distributed) is:
$$\prod_{i=1}^n  \zeta_i^{s_i}(1-\zeta_i)^{1-s_i} = \prod_{i=1}^n \left ( \frac{\zeta_i}{1-\zeta_i} \right )^{s_i}(1-\zeta_i).$$
The log-likelihood is:
$$\LOG(\theta;\s,X) = \sum_{i=1}^n s_i \langle X_i,\theta \rangle  - \ln(1+\exp(\langle X_i, \theta \rangle)).$$

\section{Least Singular Value of Perturbed Random Matrices }
In this section, we bound the least singular value of a random matrix perturbed by a low rank matrix. We need the following fact.
\begin{lemma} \label{l: ball}
Let $\mathbf{R}$ be a $d \times n$ random matrix with independent centered subgaussian entries with variances uniformly bounded below ($\tau$-random entries fall into this category) and with $d \geq c_{14}n$. 
For any $\z \in \R^d$,
\[ \Pr \left [\exists \x \in S^{n-1} \ \norm{\mathbf{R}\x + \z} \le c_{15}\sqrt{d} \right ] \le \exp(-c_{16} d). \]
The constants $c_{14},c_{15}$, and $c_{16}$ are all independent of $n$ and $d$. 
\end{lemma}
For $\z=0$ this follows from Proposition 2.5~\cite{RV08} and the standard estimate of the norm of a subgaussian matrix (see e.g., Proposition 2.3~\cite{RV}). The proof for a general $\z$ follows the same lines.

The lemma below gives an estimate of the smallest singular value of a perturbed random matrix.
\begin{lemma} \label{lem:randpert}
Let $\mathbf{R}$ be a $d \times n$ random matrix with independent centered subgaussian entries with variances uniformly bounded below and $d \geq c_{14}n$. Let $D$ be a deterministic $d \times n$ matrix such that $\text{rank}(D)=K, \ \norm{D} \le d^a$, where $a>0$ is a constant. If 
\[\begin{cases}
K \le \frac{c_{17}d}{(a-1/2) \log d} & \text{if } a>1/2 , \\
K \le c_{18}d & \text{if } a \le 1/2.
\end{cases} \]
%\begin{align*} & K \le \frac{c_{17}d}{(a-1/2) \log d} \;\;\;\ \mbox{ for } a>1/2 \mbox{ or } \\  &  K \le c_{18}d \;\;\; \mbox{ for } a \le 1/2, \end{align*}
 then 
\[ \Pr [\sigma_n(\mathbf{R}+D) \le c_{19} \sqrt{d}] \le \exp(-c_{20}d). \]
The constants $c_{17},c_{18},c_{19}$, and $c_{20}$ are all independent of $n$ and $d$.
\end{lemma}
\begin{proof}
Set $Q=D B_2^n$ (where $B_2^n$ denotes the unit Euclidean ball in $\R^n$), and let $\epsilon=c_{15} \sqrt{d}/2$. By the volumetric estimate, there exists an $\epsilon$-net $\NN$ in $Q$ of cardinality at most
 \[ |\NN| \le \left( \frac{3}{\epsilon} \right)^K \norm{D} \le (c_{21}d)^{(a-1/2)K} \]
 for $a>1/2$ and $|\NN| \le c_{22}^K$ for $a \le 1/2$. Assume that there exists $\x \in S^{n-1}$ such that $\norm{(\mathbf{R}+D)\x} \le c_{15}\sqrt{d}/2$. Choose $\z \in \NN$ so that $\norm{D\x-\z} <\epsilon$. Then
\begin{align*} \norm{\mathbf{R}\x+\z} & \le \norm{(\mathbf{R}+D)\x} +\norm{\z-D\x} \\ 
& < c_{15} \sqrt{d}/2+\epsilon=c_{15}\sqrt{d}. 
\end{align*}
Lemma~\ref{l: ball} and the union bound yield
\begin{align*}
& \Pr [\exists \x \in S^{n-1} \  \norm{(\mathbf{R}+D)\x}_2 \le c_{15} \sqrt{d}/2 ] \\
& \le  \Pr [\exists \x \in S^{n-1} \exists z \in \NN \ \norm{\mathbf{R}\x+\z}_2\le c_{15} \sqrt{d}] \\
& \le |\NN| \cdot \exp(-c_{16} d)\\
& \le \exp(-c_{20} d).
 \end{align*}
 The last inequality follows from the assumption on $K$. Here, $c_{21}$ and $c_{22}$ are constants independent of $n$ and~$d$. 
\end{proof}
\end{document}